\documentclass[a4paper]{amsart}
\usepackage[a4paper,margin=2.5cm]{geometry}
\usepackage{amsmath, amssymb, amsthm, mathtools, abraces, comment, enumitem, filecontents}
\usepackage[utf8]{inputenc}
\usepackage[dvipsnames]{xcolor}
\usepackage[hypertexnames=false, hyperfootnotes=false, colorlinks=true, linkcolor=blue, citecolor=purple, filecolor=magenta, urlcolor=cyan, unicode, linktocpage=true]{hyperref}

\newtheorem{theorem}{Theorem}[section]     
\newtheorem{proposition}[theorem]{Proposition} 
\newtheorem{lemma}[theorem]{Lemma} 
\newtheorem{corollary}[theorem]{Corollary}

\theoremstyle{definition}

\theoremstyle{remark}
\newtheorem{remark}[theorem]{Remark}

\DeclareMathOperator{\res}{Res}

\DeclareMathOperator{\diag}{diag}

\def\Z{\mathbb{Z}}	
\def\C{\mathbb{C}}	
\def\R{\mathbb{R}}	
\renewcommand{\leq}{\leqslant} 		
\renewcommand{\geq}{\geqslant}

\def\ii{{\mathrm i}} 		
\def\cC{\mathcal{C}}
\def\cI{\mathcal{I}}
\def\cU{\mathcal{U}}
\def\cV{\mathcal{V}}
\def\cH{\mathcal{H}}

\def\bla{{\bar\lambda}}
\def\bem#1\enm{\begin{pmatrix}#1\end{pmatrix}}

\numberwithin{equation}{section}

\begin{document}

\title[Stokes phenomenon and Frobenius manifolds]
{Infinite-dimensional Dubrovin-Frobenius manifolds \\ and the Stokes phenomenon}
\author{Guido Carlet}
\address[G. Carlet]{Institut de Mathématiques de Bourgogne, UMR 5584 CNRS, Université Bourgogne Franche-Comté, F--2100 Dijon, France} 
\email{guido.carlet@u-bourgogne.fr}
\author{Francisco Hern{\'a}ndez Iglesias}
\address[F. Hern{\'a}ndez Iglesias]{Korteweg-de Vriesinstituut voor Wiskunde, 
Universiteit van Amsterdam, Postbus 94248,
1090GE Amsterdam, Nederland \newline Institut de Mathématiques de Bourgogne, UMR 5584 CNRS, Université Bourgogne Franche-Comté, F--2100 Dijon, France}
\email{f.hernandeziglesias@uva.nl}
\begin{abstract} 
We study the Dubrovin equation of the infinite-dimensional 2D Toda Dubrovin-Frobenius manifold at its irregular singularity. We first revisit the definition of the canonical coordinates, proving that they emerge naturally as generalized eigenvalues of the operator of multiplication by the Euler vector field. We then show that the formal solutions to the Dubrovin equation with exponential type behaviour at the irregular singular point are not uniquely determined by their leading order, but instead depend on an infinite number of parameters, contrary to what happens in the finite-dimensional case. 
Next, we obtain a large family of solutions to the Dubrovin equation given by integrals along the unit circle of certain combinations of the superpotentials. Observing that such a family is not complete and has trivial monodromy, we study a larger family of weak solutions obtained via Borel resummation of some distinguished formal solutions. These resummed solutions naturally appear in monodromy-related pairs, finally allowing us to compute the infinite-dimensional analogue of the Stokes matrices.
\end{abstract}
\maketitle
\tableofcontents

\section{Introduction}

Dubrovin--Frobenius manifolds were introduced by B.~Dubrovin in~\cite{Dub96} to provide a coordinate-free description of the WDVV associativity equations~\cite{Wit90, DVV91} of two-dimensional topological field theory. 
While on the one hand Dubrovin--Frobenius manifolds provide the leading invariant in the reconstruction of higher-genus generating functions of several enumerative objects, on the other hand they have proven valuable in the classification and study of a large class of integrable hierarchies with one  spatial variable~\cite{DZ01}.

The program of extending the tools of Dubrovin--Frobenius manifold theory to integrable hierarchies in two spatial variables, that is, $2+1$ integrable systems, started in~\cite{CDM11} with the definition of an infinite-dimensional Dubrovin--Frobenius manifold $M_0$ associated with the dispersionless limit of the bi-Hamiltonian structure~\cite{carletHamiltonianStructuresTwodimensional2005} of the 2D Toda lattice due to Ueno and Takasaki~\cite{UT84}. 
In~\cite{CM15} the Dubrovin equation of $M_0$ was derived and studied, in particular by obtaining a Levelt basis of solutions near its regular singular point at $\zeta \sim 0$. This yields a canonical basis of Hamiltonian densities for the principal hierarchy of $M_0$, which constitutes a non-trivial extension of the dispersionless 2D Toda lattice. 

In recent years, several other examples of infinite-dimensional Dubrovin--Frobenius manifolds have been constructed. In~\cite{WZ14} a family $M^{n,m}_0$ of infinite-dimensional Dubrovin--Frobenius manifolds, all of them underlying the dispersionless 2D Toda lattice and coinciding with $M_0$ for $n=m=1$, was defined. A similar infinite family for the dispersionless two-component BKP hierarchy was discussed in~\cite{WX12}. Other remarkable examples are the infinite-dimensional Dubrovin--Frobenius manifold associated with the dispersionless KP hierarchy, defined in~\cite{Rai12}, and a family of infinite-dimensional Dubrovin--Frobenius manifolds underlying the Whitham hierarchy, recently obtained in~\cite{mwz21}.

The existence of a theory in full genera associated with these infinite-dimensional Dubrovin--Frobenius manifolds is still not clear. In this direction, a partial cohomological field theory of infinite rank has been recently defined in~\cite{buryakModuliSpacesResidueless2021}. Its associated Hamiltonian integrable hierarchy, in a certain reduction, has been shown to coincide with the KP hierarchy. 

In this paper we continue the study of the Frobenius manifold $M_0$ associated with the 2D Toda hierarchy. 

First, we revisit the definition of the canonical coordinates introduced in~\cite{CDM11}, showing that the continuous family $u_p$ has to be supplemented by a finite number of discrete coordinates $u_i$, $\bar{u}_i$ given by the critical values of the two superpotentials $\lambda$, $\bar\lambda$, in analogy to the usual description of canonical coordinates for finite-dimensional Frobenius manifolds given by a superpotential. To give a better justification for the somewhat ad hoc definition of the canonical coordinates $u_p$, we study the spectrum of the operator $\cU$ of multiplication by the Euler vector field. We show that the continuous canonical coordinates $u_p$ coincide with the generalized eigenvalues of $\cU$, while the standard eigenvalues are given by the critical values $u_i$, $\bar{u}_i$.
To give a more accurate and rigorous description of the tangent and cotangent spaces to $M_0$, here we make a distinction between the cotangent space and its representable (via the metric) subspace. This is necessary to deal with the non-representable differentials of several basic functionals on $M_0$, including those of the canonical coordinates. 

We then consider the Dubrovin equation at its irregular singularity at $\zeta \sim\infty$. We reformulate it as an equation on the cotangent space to $M_0$, rather than on its representable subspace, to allow for sufficiently large families of solutions. We study the formal solutions of the Dubrovin equation at the irregular singularity, remarkably finding that such formal solutions are not uniquely determined by their leading order, unlike in the finite-dimensional case, but depend on a large set of parameters. 

Our final aim is to describe the Stokes phenomenon for the irregular singularity of the Dubrovin equation and, in particular, to compute its Stokes matrices. We obtain an infinite family of solutions given by integrals along the unit circle and compute their asymptotics. We are however faced with the problem that such a family has trivial monodromy around $\zeta \sim \infty $ and cannot be considered as the analogue of a fundamental solution in the finite-dimensional case or, in other words, it is not complete. To solve this problem, we apply the theory of resurgent functions to certain formal solutions for which we have an explicit description, namely those obtained as asymptotic series from the integral solutions. What we find in the resummation process is a large family of solutions which are nevertheless weak, i.e., they do not extend to linear functionals defined on the whole tangent space. For such a family, we explicitly compute the Stokes matrices. For simplicity, this last part of the paper is conducted restricting to a two-dimensional locus in $M_0$ where the superpotentials have a particularly simple form.

\subsection*{Organization of the paper} 

In Section~\ref{sec:2DFM} we recall the definition of the 2D Toda Dubrovin--Frobenius manifold $M_0$ given in~\cite{CDM11, CM15}. 
In Section~\ref{sec:canonical} we revisit the canonical coordinates and prove they coincide with the (generalized) eigenvalues of the operator $\cU$ of multiplication by Euler vector field. 
In Section~\ref{sec:Dubrovinequation} we derive the Dubrovin equation on the cotangent spaces. 
In Section~\ref{sec:formalsolutions} we find the formal solutions to the Dubrovin equation at $\infty$. 
In Section~\ref{sec:analyticsolutions} we study an infinite, albeit incomplete, family of integral solutions to the Dubrovin equation with suitable asymptotic expansions at $\infty$. 
Finally, in Section~\ref{sec:resurgenceandstokes} we apply the resurgence procedure to the formal solutions which arise as asymptotic expansions of the integral solutions, obtaining this way a family of weak solutions parametrized by the unit circle $S^1$. These solutions appear naturally in monodromy-related pairs, allowing us to study the Stokes phenomenon in a similar fashion to the finite-dimensional case. The section ends with the explicit computation of the infinite-dimensional analogue of the Stokes matrices.

\subsection*{Acknowledgments} 

G.~C. would like to acknowledge many hours of conversation with B.~Dubrovin about the topics covered in this works.
G.~C. is supported by the ANER grant ``FROBENIUS'' of the Region Bourgogne-Franche-Comt\'e. The IMB receives support from
the EIPHI Graduate School (contract ANR-17-EURE-0002). F.~H.~I. is supported by the Netherlands Organization for Scientific Research.

\section{The 2D Toda Dubrovin-Frobenius manifold} 
\label{sec:2DFM}

In this section, we recall the definition and some properties of the 2D Toda Dubrovin--Frobenius manifold from~\cite{CDM11, CM15}.

\subsection{The manifold $M$ and its tangent bundle}

Let $D_0$ be the closed unit disc in the Riemann sphere, $D_{\infty}$ the closure of its complement and $S^1 = D_0 \cap D_\infty$ the unit circle. For a compact subset $K$ of the Riemann sphere, we denote by $\cH (K)$ the space of holomorphic functions on $K$, i.e., functions which extend holomorphically to an open neighborhood of $K$. 

We define the infinite-dimensional manifold $M$ as the affine space
\begin{align}
M = \{ (\lambda(z), \bar{\lambda} (z )) \in z \mathcal{H}(D_\infty) \oplus \frac{1}{z} \mathcal{H} (D_0) \  | \ \lambda(z) = z + O(1) \}.
\end{align}
A point $\hat{\lambda} = (\lambda(z), \bar{\lambda}(z)) \in  M$ can be represented by  the Laurent series at $\infty$ and $0$ of its components
\begin{align}
\lambda(z) = z + \sum_{k \leq 0} u_k z^k, \qquad \bar{\lambda}(z) = \sum_{k \geq -1} \bar{u}_k z^k.
\end{align}
We identify the tangent space at a point $\hat{\lambda}$ with the vector space underlying the affine space $M$
\begin{align}
T_{\hat{\lambda}} M = \mathcal{H} (D_\infty) \oplus \frac{1}{z} \mathcal{H} (D_0).
\end{align}

\subsection{The manifold $M_0$}

We define $M_0$ 
as the open subset of $M$ consisting of the pairs $(\lambda(z), \bar{\lambda}(z) )$ satisfying the following conditions:
\begin{enumerate}[label=(T\arabic*)]
\item The leading coefficient $\bar{u}_{-1}$ of $\bar\lambda(z)$ is nonzero.
\label{item:T1}
\item The derivative of $w(z) := \lambda(z) + \bar{\lambda}(z)$ does not vanish on $S^1$.
\label{item:T2}
\item The curve parameterized by $w(z)$ for $z\in S^1$ is positively oriented, non-selfintersecting and encircles the origin $w=0$.
\label{item:T3}
\item The map $\sigma(z) := \frac{\lambda^{\prime}(z)}{\lambda^{\prime}(z) + \bar{\lambda}'(z)}$ has non-vanishing derivative on $S^1$. \label{item:T5}
\item The functions $\lambda'(z)$, $\bar{\lambda}'(z)$ are non-vanishing for $z\in S^1$; equivalently, the curve $\sigma: S^1 \to \C$ does not pass through the points $0$ and $1$.
\label{item:T6}
\end{enumerate}

\begin{remark}
These conditions were introduced in the literature in different places~\cite{CDM11, WZ14, CM15} mainly to avoid non-generic cases and to simplify some of the definitions and the proofs. Conditions~\ref{item:T2} and~\ref{item:T3} are used in the definition of the metric and the flat coordinates. Conditions~\ref{item:T5} and~\ref{item:T6} are used in the definition of canonical coordinates and in the computation of the spectrum of the operator $\cU$. 
\end{remark}

\subsection{The $w$-coordinates }

Sometimes it is more convenient to represent $M_0$ as a two-dimensional bundle over the space $M_{\textrm{red}} \subset \cH(S^1)$ of parameterized simple analytic curves:
\begin{align}
M_0&\longrightarrow M_{\textrm{red}} \oplus \mathbb{C} \oplus \mathbb{C} \nonumber \\
(\lambda(z), \bar{\lambda}(z)) &\longmapsto (w(z), v, u), \notag
\end{align}
where $w(z) = \lambda(z) + \bar{\lambda}(z)$, $v = \bar{u}_0 = (\bar{\lambda})_0$ and $e^u = \bar{u}_{-1} = (\bar{\lambda})_1$. The map can be inverted by
\begin{align}
\lambda(z) = w_{\leq 0} (z) + z - v -e^u z^{-1}, \qquad \bar{\lambda}(z) = w_{\geq 1} (z) - z +v +e^u z^{-1}.
\end{align}
We refer to the triples $(w(z), v, u)$ as $w$-coordinates. In these coordinates the tangent vectors are represented as elements of $\cH(S^1)\oplus \C^2$ via the map
\begin{equation}
\begin{split}
T_{\hat{\lambda}} M = \cH (D_{\infty}) \oplus \frac{1}{z} \cH (D_0) &\longrightarrow \cH(S^1) \oplus \mathbb{C}^2  \\
(X(z), \bar{X}(z)) &\longmapsto (W(z), X_v, X_u),
\end{split}
\end{equation}
where
\begin{align}
&W(z) = X(z) + \bar{X}(z), \quad X_v = \bar{X}_0, \quad X_u = e^{-u} \bar{X}_{-1}, \label{eqn:pairstows}\\
&X(z) = W_{\leq 0}(z) - X_v - e^u X_u z^{-1}, \quad \bar{X}(z) = W_{\geq 1}(z) + X_v + e^u X_u z^{-1}.
\end{align}

\begin{remark}
Recall that the projections $( \hspace{1mm} )_{\geq p} : \mathcal{H}(S^1) \rightarrow z^p \mathcal{H} (D_0)$, $(  \hspace{1mm})_{\leq p-1} : \mathcal{H}(S^1) \rightarrow z^{p-1} \mathcal{H} (D_\infty)$ and $( \hspace{1mm} )_{p} : \mathcal{H}(S^1) \rightarrow \mathbb{C}$ are defined by
\begin{align}
(f)_{\geq p}(z) &= \sum_{k \geq p} f_k z^k = \frac{z^p}{2 \pi \ii} \oint_{|z| < |w|} \frac{w^{-p} f(w)}{w-z} dw,  \\
(f)_{\leq p-1}(z)&= \sum_{k \leq p-1} f_k z^k = -\frac{z^p}{2 \pi \ii} \oint_{|z| > |w|} \frac{w^{-p} f(w)}{w-z} dw,  \\
(f)_p &= f_p = \frac{1}{2 \pi \ii} \oint_{|z| = 1} f(z) z^{-p} \frac{dz}{z} ,
\end{align}
where $f(z) = \sum_{k \in \mathbb{Z}} f_k z^k$ and $p \in \mathbb{Z}$. 
\end{remark}

\subsection{The metric and the contangent bundle} 
 
On the tangent spaces we define a symmetric non-degenerate bilinear form $\eta$, called the metric, by 
\begin{equation}
\eta(\hat{X}, \hat{Y}) = \frac{1}{2 \pi \ii} \oint_{|z|=1} \frac{X(z) Y(z) }{z^2 w^{\prime}(z)} dz + X_v Y_u + X_u Y_v,  \label{eqn:metricdefn}
\end{equation}
where $\hat{X}, \hat{Y} \in T_{\hat{\lambda}} M$ are represented as triples in $\cH(S^1) \oplus \mathbb{C}^2$.
By explicitly constructing the flat coordinates, it was proved in~\cite{CDM11} that the metric $\eta$ is flat. 

The cotangent space $T^*_{\hat{\lambda}} M$ is defined as the algebraic dual to the tangent space, i.e., as the space $(T_{\hat{\lambda}}M)^*$ of all linear functionals on $T_{\hat{\lambda}}M$.
The metric defines an injection $\eta_*$ of $T_{\hat{\lambda}} M$ into $T^*_{\hat{\lambda}} M$ by
\begin{equation}
\hat{X} \mapsto \eta_* (\hat{X}) = \eta(\hat{X}, \cdot ).
\end{equation}
A cotangent vector $\xi \in T^*_{\hat{\lambda}} M$  that is in the image of $\eta_*$ is called representable, and we denote $\xi \in T^*_{\hat{\lambda}} M^{\mathrm{rep}}$.

\begin{remark}
In this work we take a rather different approach to the cotangent bundle compared to~\cite{CDM11, CM15}. This is motivated by the fact that we need to consider functionals on $M_0$ whose differentials are not representable. For example, the differentials $d\lambda(p)$, $d\bar\lambda(p)$ and $du_p$ are not representable. 
\end{remark}

\subsection{The associative product}

The product on the tangent spaces is defined by 
\begin{align} \label{eq:protan}
\hat{X} \cdot \hat{Y} = 
&\Bigg( X(z) \left( Y_{>0}(z) - (z w^{\prime}(z))_{>0} \frac{Y(z)}{z w^{\prime}(z)} + \frac{Y(z)}{w^{\prime}(z)} + \frac{e^u}{z} \left( \frac{Y(z)}{z w^{\prime}(z)} + Y_u \right) + Y_v \right)    \\ 
&\quad + zw^{\prime}(z) \bigg( \left( X_{>0}(z) \frac{Y(z)}{z w^{\prime}(z)} \right)_{<0} - \left( X_{\leq 0} (z) \frac{Y(z)}{z w^{\prime}(z)} \right)_{\geq 0}   \nonumber \\ 
&\qquad \qquad \qquad  + \frac{e^u}{z} X_u \left( \frac{Y(z)}{z w^{\prime}(z)} + Y_u \right)  \left. + X_v \frac{Y(z)}{z w^{\prime}(z)} \right), \nonumber \\
&\left(  e^u (X(z) + zw^{\prime}(z) X_u) \left( \frac{Y(z)}{z w^{\prime}(z)} + Y_u \right) \right)_1 - e^u X_uY_u + X_v Y_v,  \nonumber \\
&\left( X(z) \frac{Y(z)}{z w^{\prime}(z)} \right)_0 + X_u Y_v + X_v Y_u \Bigg) \notag
\end{align}
for $\hat X, \hat Y \in T_{\hat{\lambda}}M$ represented as triples in $\cH(S^1) \oplus \mathbb{C}^2$.
It was proved in~\cite{CDM11} that the product is commutative, associative, with a unit vector field given by $e = (-1, 1)$ or, equivalently, by $e = (0, 1, 0)$.
Moreover, it is compatible with the metric $\eta$, namely
\begin{equation}
\eta(\hat X \cdot \hat Y, \hat Z ) = \eta( \hat X , \hat Y \cdot \hat Z ) 
\end{equation}
for any $\hat X, \hat Y, \hat Z \in T_{\hat{\lambda}}M$. If follows that $\eta(\hat X, \hat Y) =  \xi( \hat X \cdot \hat Y )$ for $\xi \in T^*_{\hat{\lambda}} M^{\mathrm{rep}}$, with $\xi = \eta_*(e) = du$.

\begin{remark}
Expression~\eqref{eq:protan} for the product of tangent vectors corrects a sign mistake in the literature, cf.~\cite[Lemma 15]{CM15}.
\end{remark}

Finally, the Euler vector field is defined by 
\begin{equation}
E = (\lambda(z) - z \lambda^{\prime}(z), \bar{\lambda}(z) - z \bar{\lambda}^{\prime}(z)), \qquad \text{or} \qquad 
E = (w(z) - z w^{\prime}(z), v, 2).
\end{equation}

In~\cite{CDM11} it is proved that 
\begin{theorem}
$(M_0, \eta,\cdot, e, E)$ is an infinite-dimensional Dubrovin--Frobenius manifold of charge $d = 1$.
\end{theorem}

\subsection{The operators $\mathcal{U}$ and $\mathcal{V}$}

The operator $\mathcal{U}: T_{\hat{\lambda}}M \rightarrow T_{\hat{\lambda}}M$ of multiplication by the Euler vector field is defined on each tangent space as $\mathcal{U} (\hat{X}) = E \cdot \hat{X}$. Using~\eqref{eq:protan} one obtains 
\small
\begin{align}
\label{eqn:defnU}
\mathcal{U}(\hat{X}) = 
&\Bigg( (w(z) - zw^{\prime}(z)) \left( X_{>0}(z) - (z w^{\prime}(z))_{>0} \frac{X(z)}{z w^{\prime}(z)} + \frac{X(z)}{w^{\prime}(z)} + \frac{e^u}{z} \left( \frac{X(z)}{z w^{\prime}(z)} + X_u \right) + X_v \right)+  \\ 
&\quad  + zw^{\prime}(z) \Bigg( \left( (w(z) - zw^{\prime}(z))_{>0} \frac{X(z)}{z w^{\prime}(z)} \right)_{<0} \notag 
  - \left( (w(z)- zw^{\prime}(z))_{\leq 0}  \frac{X(z)}{z w^{\prime}(z)} \right)_{\geq 0} \nonumber  \\  
&\qquad \qquad \qquad    + 2 \frac{e^u}{z}  \left( \frac{X(z)}{z w^{\prime}(z)} + X_u \right) + v \frac{X(z)}{z w^{\prime}(z)} \Bigg), \nonumber \\
&\left(  e^u (w(z) + zw^{\prime}(z)) \left( \frac{X(z)}{z w^{\prime}(z)} + X_u \right) \right)_1 - 2e^u X_u + v X_v, \nonumber \\
&\left( (w(z) - z w^{\prime}(z)) \frac{X(z)}{z w^{\prime}(z)} \right)_0 + 2 X_v + v X_u \Bigg) \notag.
\end{align}
\normalsize
The grading operator $\mathcal{V}: T_{\hat{\lambda}}M \rightarrow T_{\hat{\lambda}}M$ is defined as 
\begin{align}
\mathcal{V} = \frac12 - \nabla E, \label{eqn:mathcalVdefn}
\end{align}
where $\nabla$ is the Levi-Civita connection of the metric $\eta$. Explicitly, see~\cite{CM15}, it is given by
\begin{align}
\mathcal{V}(\hat{X}) = \left( - \frac{X(z)}{2} + z \partial_z \left( X(z) \frac{w(z)}{z w^{\prime}(z)} \right), - \frac{X_v}{2}, \frac{X_u}{2} \right). \label{eqn:defnV}
\end{align}

\subsection{At a special point}

To simplify computations, we will specialize certain constructions to a two-dimensional submanifold of $M_0$ given by the points $\hat{\lambda}_0$ of the form 
\begin{align}
\lambda_0(z) = z - v - e^u z^{-1}, \qquad \bar{\lambda}_0 (z) = v + e^u z^{-1} 
\end{align} 
or, written as a triple, 
\begin{align}
\hat{\lambda}_0 = (z, v, u). 
\end{align}
Notice that conditions~\ref{item:T1}-\ref{item:T6} are satisfied if $|e^u| \not=1$.

At $\hat{\lambda}_0$ the operators $\cU$ and $\cV$ have the simpler form
\begin{align}
\mathcal{U} (\hat{X}) &= \left( (v + 2e^u z^{-1}) X(z) + 2e^u X_u, \ 
2e^u X_1 + v X_v, \ 2 X_v + v X_u \right), \label{eqn:Uatsimplepoint} \\ 
\mathcal{V}(\hat{X}) &= \left( - \frac{X(z)}{2} + z X^{\prime}(z), - \frac{X_v}{2}, \frac{X_u}{2} \right). \label{eqn:Vatsimplepoint}
\end{align}

\section{Spectrum of $\cU$ and canonical coordinates} \label{sec:canonical}

In this section, we compute the spectrum of the operator $\cU$ at an arbitrary point of the Dubrovin--Frobenius manifold and we show that the generalized eigenvalues correspond to the continuous canonical coordinates introduced in~\cite{CDM11}, while the discrete spectrum is given by the critical values of $\lambda$ and $\bar\lambda$.

\subsection{Canonical coordinates}

For a semisimple finite dimensional Dubrovin--Frobenius manifold with superpotential $\lambda(z)$, the canonical coordinates are typically given by the critical values of $\lambda(z)$. In the case of the infinite dimensional Dubrovin--Frobenius manifold $M_0$, however, it is not immediately clear what should take the place of the critical values, since one expects an infinite number of canonical coordinates and, instead of a single superpotential, there are two: $\lambda(z)$ and $\bar{\lambda}(z)$. 

In~\cite{CDM11} it was suggested to consider the following linear combination of the two superpotentials
\begin{align}
\lambda_\sigma(z) = \sigma \bar{\lambda}(z) + (\sigma-1) \lambda(z) \in \cH(S^1)
\end{align}
for a parameter $\sigma \in \C$. One should then look for the critical points of $\lambda_\sigma(z)$ that are located on $S^1$. 
The condition $\lambda_\sigma^{\prime}(z) = 0$ for $z\in S^1$ defines a curve $\Sigma = \{ \sigma(z)| \ z \in S^1 \}$, parameterized by
\begin{align} \label{eqn:curvesigma}
\sigma(z) = \frac{\lambda^{\prime}(z)}{\lambda^{\prime}(z) + \bar{\lambda}^{\prime}(z)} \in \cH(S^1),
\end{align}
which is holomorphic on $S^1$ as the denominator is non-vanishing for $\hat{\lambda} \in M_0$, and is non-singular, i.e. $\sigma'(z)\not=0$, if and only if
\begin{equation} \label{eq:nondeggg}
\lambda'(z) \bar\lambda''(z) - \lambda''(z) \bar\lambda'(z) \not= 0 
\end{equation}
for $z \in S^1$. For non self-intersecting $\Sigma$, we define the (continuous part of the) canonical coordinates at the point $\hat{\lambda}$ as the set of critical values
\begin{equation}
u_\sigma = \lambda_\sigma(z(\sigma))
\end{equation}
for $\sigma \in \Sigma$, where $z(\sigma): \Sigma \to S^1$ is the inverse of $\sigma(z)$, which is a critical point of $\lambda_\sigma(z)$. Since $\Sigma$ is parameterized by $z \in S^1$, we might as well index these coordinates by $p\in S^1$, denoting $u_p = u_{\sigma(p)} = \lambda_{\sigma(p)}(p)$. 

In the following, we show that this seemingly ad hoc definition of canonical coordinates emerges naturally from the spectrum of the operator $\cU$. Indeed, the generalized eigenvalues of $\cU$ are exactly given by the canonical coordinates defined above. 

The operator $\cU$ turns out to also have standard eigenvalues, which are given by the critical values of the superpotentials $-\lambda(z)$ and $\bar{\lambda}(z)$ on their respective domains of definition, $D_\infty$ and $D_0$. More precisely, consider a point of $M_0$ at which $\lambda(z)$, resp. $\bar{\lambda}(z)$, has $n$, resp. $\bar{n}$, critical points in $D_\infty$, resp. $D_0$. 
We define the following critical values:
\begin{align} \label{eq:uub}
&u_i = -\lambda(z_i), \qquad 
&\lambda^{\prime}(z_i) = 0, \qquad 
&z_i \in D_\infty, \qquad &i = 1, \dots, n, \\
 \label{eq:uub1}
&\bar{u}_i = \bar{\lambda}(\bar{z}_i), \qquad 
&\bar{\lambda}^{\prime}(\bar{z}_i) = 0, \qquad 
&\bar{z}_i \in D_0, \qquad &i = 1, \dots, \bar{n}.
\end{align}
The canonical coordinates on $M_0$ are given by the set of all critical values as defined above: 
\begin{equation}
\{ u_p, u_i, \bar{u}_j \}_{p \in S^1, i=1, \dots, n, j = 1, \dots, \bar{n}}.
\end{equation}

The differentials of the discrete canonical coordinates $u_i$, $\bar{u}_j$ are
\begin{equation}
du_i = - d\lambda(z_i), \qquad d\bar{u}_i = d\bar{\lambda}(\bar{z}_i), \label{eqn:diffdiscrete}
\end{equation}
which can be represented as vectors in $T_{\hat{\lambda}} M$ via the injection $\eta_*$ as follows
\begin{align} \label{eqn:duirep}
du_i = \left( zw^{\prime}(z) \frac{z_i}{z - z_i}, \frac{e^u}{z_i}, 1 \right), \qquad 
d\bar{u}_i = \left( z w^{\prime}(z) \frac{\bar{z}_i}{z- \bar{z}_i}, \frac{e^u}{\bar{z}_i}, 1 \right).
\end{align}
We will show below that these differentials are actually the eigenvectors corresponding to the eigenvalues $u_i$ and $\bar{u}_j$ of $\cU$. 

It turns out that the generalized eigenvectors of $\cU$, corresponding to the continuous family of canonical coordinates $u_p$, are given by 
\begin{align}
du_p := d\lambda_{\sigma}(p)|_{\sigma = \sigma(p)} 
=  (\sigma(p) - 1) d\lambda(p) + \sigma(p) d \bar{\lambda}(p),
\label{eqn:diffcanonical}
\end{align}
for $p\in S^1$. 

\begin{remark}
Notice that in the previous definition we have slightly abused the notation, since the last formula does not represent the differential of $u_p$, but the differential of $\lambda_\sigma(z)$ for fixed $\sigma$, later evaluated at $\sigma=\sigma(p)$. This is consistent with the fact that, as in the case of discrete canonical coordinates, the critical point should be allowed to vary as we differentiate along the Dubrovin--Frobenius manifold, but on the contrary it would be fixed at a point of $S^1$ if we differentiated directly $u_p$. 
\end{remark}

\begin{remark} 
The formula for the continuous canonical coordinates might be understood as the Legendre transform of the function $\lambda(w) = \lambda(z(w))$, where $z(w)$ is the inverse of the function $w(z) = \lambda(z) +\bar{\lambda}(z)$ defined on $S^1$. Denote by $w(\sigma)$ the inverse of 
\begin{equation}
\sigma(w) = \frac{\partial \lambda}{\partial w}(w) = \lambda'(z(w)) z'(w) =\left.  \frac{\lambda'(z)}{\lambda'(z) +\bar{\lambda}'(z)} \right|_{z=z(w)}.
\end{equation}
The Legendre transform of $\lambda(w)$ is indeed
\begin{align}
 \sigma w(\sigma) - \lambda(w(\sigma)) 
&= \sigma	w(z(\sigma)) - \lambda(w(z(\sigma))) =\\
&=\left[ \sigma (\lambda(z) + \bar{\lambda}(z) ) - \lambda(z) \right]_{z=z(\sigma)} 
= \lambda_\sigma (z(\sigma)) = u_\sigma ,
\end{align}
where $z(\sigma)$ is the inverse of~\eqref{eqn:curvesigma}. 
\end{remark}

\subsection{Spectrum of $\cU$}

Let us consider the operator $\cU$ of multiplication by the Euler vector field $E$, see~\eqref{eqn:defnU}, at an arbitrary point $\hat{\lambda}$ in $M_0$:
\begin{equation}
\cU : T_{\hat{\lambda}} M \to T_{\hat{\lambda}} M. 
\end{equation}
The generalized spectrum of the operator $\cU$ is defined as the spectrum of the transpose 
\begin{equation}
\cU^* : T^*_{\hat{\lambda}} M \to T^*_{\hat{\lambda}} M,
\end{equation}
defined by $<\cU^* \xi , \hat X> = < \xi, \cU \hat X>$ for all $\hat X \in T_{\hat{\lambda}} M$.
Explicitly, we say that $\xi \in  T^*_{\hat{\lambda}} M$ is a generalized eigenvector corresponding to the generalized eigenvalue $\mu$ if 
\begin{equation}
< \xi, \cU \hat X > = \mu <\xi ,  \hat X > 
\end{equation}
for all $\hat X \in T_{\hat{\lambda}} M$.
Since $\cU$ is symmetric w.r.t. the metric $\eta$, a standard eigenvector $\hat{X} \in T_{\hat{\lambda}} M$  with eigenvalue $\mu$ is mapped by the injection $\eta_*$ to a generalized eigenvector for the same eigenvalue $\mu$.

Notice that a family $E \subset T^*_{\hat{\lambda}} M$ of cotangent vectors defines a map from $T_{\hat{\lambda}} M$ to the space of functions over $E$. We say that $E$ is complete if such map is injective, i.e., it defines an isomorphism of $T_{\hat{\lambda}} M$ with the space of functions $E'$ given by its image. 

\begin{proposition}  \label{prop:canonical}
At an arbitrary point $\hat{\lambda}$ of $M_0$ the spectrum of the operator $\cU$ is given by 
\begin{enumerate}
\item the eigenvalues $u_i$ with eigenvectors $du_i$ for $i=1, \dots ,n$,
\item the eigenvalues $\bar u_j$ with eigenvectors $d\bar u_j$ for $j=1, \dots ,\bar n$, and
\item the generalized eigenvalues $u_p$ with generalized eigenvectors $d u_p$ for $p\in S^1$.
\end{enumerate}
Moreover, the set of all eigenvectors $\{du_p, du_i, d\bar{u}_j\}$ is a complete family in $T^*_{\hat{\lambda}} M$.
\end{proposition}

Actually, the completeness of the set of eigenvectors is realized via an explicit isomorphism
\begin{align} \label{eq:isoooo}
\Psi : T_{\hat{\lambda}} M  &\longrightarrow \cH(S^1) \oplus \C^n \oplus \C^{\bar{n}} \nonumber \\
\hat{X} &\longmapsto (\langle du_z, \hat{X} \rangle, \langle du_i, \hat{X} \rangle, \langle d\bar{u}_j, \hat{X} \rangle ). 
\end{align}

\begin{corollary} \label{cor:representation}
The operator $\cU$ in the representation given by $\Psi$, i.e. $U := \Psi \mathcal{U} \Psi^{-1}$, is diagonal 
\begin{align}
U: \cH (S^1) \oplus \C^n \oplus \C^{\bar{n}} &\longrightarrow \cH (S^1) \oplus \C^n \oplus \C^{\bar{n}} \nonumber \\
\hat{Y} = (Y(z), \{ Y_i \}_{i = 1, \dots, n}, \{ \bar{Y}_j \}_{j = 1, \dots, \bar{n}}) &\longmapsto U (\hat{Y}) = (u_z Y(z), \{ u_i Y_i \}_i, \{ \bar{u}_j \bar{Y}_j \}_j ).\label{eqn:diagU}
\end{align}
\end{corollary}

We now proceed to prove Proposition~\ref{prop:canonical} first by an explicit approach at the special point in the following section, then in the general case in \S\ref{sec:proooof}. In \S\ref{sec:keyyy} we prove a key lemma that will also be used in later sections.

\subsection{Proof at the special point} 
\label{sec:canonicalspecial} 

At the special point $\hat{\lambda}_0 = (z, v, u)$, the operator $\cU$ takes the simpler form~\eqref{eqn:Uatsimplepoint}. This allows us to give an explicit proof of the proposition. It is evident in this proof that the formula for the canonical coordinates emerges from and is uniquely determined by the form of the operator $\cU$. The first part of Proposition~\ref{prop:canonical} can be restated as

\begin{lemma}
	The operator $\cU$ acting on $T_{\hat{\lambda}_0} M \cong \cH(S^1) \oplus \C^2$ has the following eigenvalues and eigenvectors
	\begin{align}
	u_{\pm} = v \pm 2 \ii e^{u/2}, \qquad du_{\pm} = ( \pm z (z \pm \ii e^{u/2})^{-1}, \mp 1, \ii e^{-u/2}  ),
	\end{align}
	iff $|e^u| > 1$ and the following generalized eigenvalues and eigenvectors
	\begin{align}
	u_p = v + \frac{2e^u}{p}, \qquad \langle du_p, \hat{X} \rangle = \frac{e^u}{p^2} X(p) + X_{\geq 1} (p) + X_v + \frac{e^u}{p} X_u,
	\end{align}
	for $p \in S^1$.
\end{lemma}

\begin{proof}
First, let us compute the eigenvalues and eigenvectors. The equation $\mathcal{U} (\hat{X}) = \mu \hat{X}$ takes the explicit form
\begin{align}
&(v + 2 e^u z^{-1}) X (z) + 2 e^u X_u = \mu X(z),  \\
&2e^u X_1 + v X_v = \mu X_v,  \\
&2 X_v + vX_u = \mu X_u .
\end{align}
For $\mu = v$ the system becomes 
\begin{align}
&z^{-1} X(z) +  X_u = 0,  \\
&X_1 = 0,  \\
&X_v = 0 .
\end{align}
The first equation implies that the only possibly non-zero coefficient of the Laurent expansion $X (z) = \sum_{k \in \mathbb{Z}} X_k z^k$ is $X_1$, which is zero by the second equation. Thus, $X_u$ also vanishes and $\hat{X} = 0$, so $\mu = v$ is not an eigenvalue. Therefore, we can assume $\mu \not= v$. 

Let $p = \frac{2 e^u}{\mu - v}$. The system becomes
\begin{align}
&(z-p) X(z) = p^3 e^{-u} z X_1, \\
& p X_1 = X_v, \\
& p^2 e^{-u} X_1 = X_u.
\end{align}
We rewrite the first equation as
\begin{align}
\frac{X(z)}{z} = \frac{p^3 e^{-u} X_1}{z - p} .
\end{align}
If $|p| = 1$, the function $X(z)$ defined as above would have a single pole at $p$, so it would not be an element of $\cH (S^1)$. Extracting the zeroth coefficient of the Laurent expansion of the left-hand side yields
\begin{align}
1 = \frac{1}{2 \pi \ii} \oint_{|z|=1} \frac{p^3 e^{-u}}{z-p} \frac{dz}{z} .
\end{align}
If $|p| < 1$, the two poles of the integrand lie inside the unit circle, so the integral vanishes and the equation admits no solutions. If $|p| > 1$, we obtain $e^u = -p^2$, which has two solutions iff $|e^u| > 1$, namely $p_{\pm} = \mp \ii e^{u/2}$, which correspond to the eigenvalues $u_{\pm}$ and the eigenvectors $du_{\pm}$.
	
Let us now compute the generalized eigenvalues. Let $\rho = \mu - v$, then the generalized eigenvalue equation takes the form
\begin{align}
\langle \omega_z, (2e^u z^{-1} - \rho) X(z) + 2e^u X_u  \rangle + \langle \omega_v,  2 e^u X_1 - \rho X_v  \rangle + \langle \omega_u, 2X_v - \rho X_u \rangle = 0,
\end{align}
for a functional $\omega = \omega_z + \omega_v + \omega_u$. If $\rho = 0$, then the previous equation becomes
\begin{align}
\langle \omega_z, 2e^u z^{-1}  X(z) + 2e^u X_u  \rangle + \langle \omega_v,  2 e^u X_1  \rangle + \langle \omega_u, 2X_v \rangle = 0 .
\end{align}
Choosing $\hat{X} = (0, X_v, 0)$ implies $\omega_u = 0$. Choosing $\hat{X} = (0, 0, X_u)$ implies $\omega_z$ is zero on constants. Then choosing $\hat{X} = (X_1 z, 0, 0)$ shows that $\langle \omega_z, 2 e^u z^{-1} X(z) \rangle = \langle \omega_z, 2 e^u X_1 \rangle = 0$ because the argument is constant, so we can conclude that $\omega_v = 0$. Finally, we choose $\hat{X} = (X(z), 0, 0)$, which shows $\omega_z = 0$. Therefore, we can assume $\rho \not= 0$.

Let $p = \frac{2e^u}{\rho}$. Substituting in the equation above, we obtain
\begin{align}
\left\langle \omega_z,  \left( \frac{p}{z} - 1 \right) X(z) + p X_u \right\rangle + \langle \omega_v,  p X_1 - X_v \rangle + \langle \omega_u, e^{-u} p X_v - X_u \rangle = 0 .
\end{align}
Choosing $\hat{X} = (0, 0, X_u)$ implies $\omega_u = p \langle \omega_z, 1 \rangle$. Choosing $\hat{X} = (0, X_v, 0)$ implies $\omega_v = e^{-u} p \omega_u = e^{-u} p^2 \langle \omega_z, 1 \rangle$. Substituting and setting $\hat{X}= (X(z),0,0)$ yields
\begin{align}
\left\langle \omega_z,  \left( \frac{p}{z} - 1 \right) X(z) \right\rangle+ e^{-u} p^3 X_1 \langle \omega_z, 1 \rangle = 0 \label{eqn:omegaz}.
\end{align}

Consider first the case $|p| \not= 1$. Multiplication by $\left( \frac{p}{z} - 1 \right) $ is then invertible in $\cH (S^1)$, so we obtain
\begin{align}
\langle \omega_z,  X(z) \rangle = - e^{-u} p^3 \left( \frac{z}{p-z} X(z) \right)_1 \langle \omega_z, 1 \rangle .
\end{align}
Clearly $\omega_z = 0$ iff $\langle \omega_z, 1 \rangle = 0$. Thus, we can assume $\langle \omega_z, 1 \rangle  \not= 0$ and, without loss of generality, take $\langle \omega_z, 1 \rangle = 1$. Setting $X(z) = 1$ gives the equation
\begin{align}
1 = -e^{-u} p^3 \left( \frac{z}{p-z} \right)_1 = -e^{-u} p^3 \left( \frac{1}{p-z} \right)_0 .
\end{align}
If $|p| < 1$, the right-hand side vanishes, so there is no solution. If $|p| > 1$, the equation becomes $p^2 = -e^u$, which admits the solutions $p_{\pm } = \mp \ii e^{u/2}$ when $|e^u|>1$. 
The generalized eigenvectors $\omega_\pm$ associated with $p_{\pm}$ have eigenvalues $u_{\pm}$ and correspond to the eigenvectors $du_{\pm}$ computed above,
more precisely $\eta_* du_\pm = -\ii e^{-u/2} \omega_\pm$. 

Finally, let us consider the case $|p| = 1$. One can check that the functional given by
\begin{align}
\langle \omega_z, X(z) \rangle = e^{-u} p^2 X_{\geq 1} (p) + X(p) \label{eqn:omegazdef}
\end{align}
satisfies $\langle \omega_z, 1 \rangle  = 1$ and equation~\eqref{eqn:omegaz}. Let us now show that it is the only solution for fixed $p$ with $|p|=1$. Let $\alpha_z$ be a solution of~\eqref{eqn:omegaz} with $\langle \alpha_z, 1 \rangle = 0$. Then $\alpha_z$ is zero on the subspace $\left( \frac{p}{z} - 1 \right) \cH (S^1)$, which is the subspace of $\cH (S^1)$ of functions vanishing at $z = p$. Therefore, 
\begin{align}
\langle \alpha_z, X(z) \rangle = \langle \alpha_z,  X(p) \rangle + \langle \alpha_z, (X(z) - X(p)) \rangle = X(p) \langle \alpha_z, 1 \rangle = 0, 
\end{align}
so $\alpha_z = 0$. Now let $\omega^{\prime}_z$ be a solution of~\eqref{eqn:omegaz} with $\langle \omega^{\prime}_z, 1 \rangle \not=0$. We can renormalize it and consider the case $\langle \omega^{\prime}_z, 1 \rangle = 1$. Then $\omega_z - \omega^{\prime}_z$ is a solution of~\eqref{eqn:omegaz} vanishing on 1, so it must be identically zero, hence $\omega^{\prime}_z = \omega_z$. The result follows by noting that $\omega = e^{-u} p^2 du_p$.
\end{proof}

\begin{remark}
Notice that in this case we have
\begin{equation}
\langle du_p, \hat{X} \rangle = \left( \left( \frac{e^u}{p^2} +1 \right) X \right)_{\geq1} + \frac{e^u}{p^2} X_{\leq0} 
+\frac{e^u}{p} (X_1 + X_u ) + (X_v + e^u X_2  ).
\end{equation}
In the case $|e^u| <1$, one can easily check that knowing $Y(p) = \langle du_p, \hat{X} \rangle$ is sufficient to reconstruct $\hat{X}$, showing completeness. However, in the case $|e^u|>1$, we also need to know $Y_{\pm} = \langle du_{\pm}, \hat{X} \rangle$ to invert~\eqref{eq:isoooo}. In Section~\ref{sec:proooof}, we will give a general formula for $\Psi^{-1}$.
\end{remark}

\subsection{A key lemma} 
\label{sec:keyyy}

The following lemma will be used in the general proof of Proposition~\ref{prop:canonical}  and also in Section~\ref{sec:analyticsolutions}.

\begin{lemma} \label{lem:keylemma}
Let $\hat{X} = (X(z), X_v, X_u) \in \cH(S^1) \oplus \C^2$. The function
\begin{align} \label{key-for}
\langle d\lambda_\sigma(z), \mathcal{U} \hat{X} \rangle - \lambda_\sigma(z) \langle d\lambda_\sigma(z), \hat{X} \rangle + z \lambda^{\prime}_\sigma(z) \left\langle d\lambda_\sigma(z), \left( \frac{w(z)}{z w^{\prime}(z)} X(z), 0, -X_u \right) \right\rangle
\end{align}
is a scalar multiple of $z \lambda_\sigma^{\prime}(z)$, namely it is equal to
\begin{align}
z \lambda_{\sigma}^{\prime}(z) \left[ \left( \left(1-\frac{w(z)}{z w^{\prime}(z)} \right) X(z) \right)_0 - X_v \right].
\end{align}
\end{lemma}

\begin{proof}
		Let us rewrite $\lambda_\sigma(z)$ as a triple in $\cH(S^1) \oplus \mathbb{C}^2$ 
		\begin{align}
		\lambda_\sigma(z) &= (\sigma - 1) w(z) + w_{\geq 1} (z) - z + v + \frac{e^u}{z}, \\ 
		\langle d\lambda_\sigma(z), \hat{X} \rangle &= (\sigma - 1) X(z) + X_{\geq 1} (z) + X_v + \frac{e^u}{z} X_u.
		\end{align}
		We proceed componentwise. Let $E(z)$ denote expression~\eqref{key-for}, and let us expand $E(z)$ for $\hat{X} = (X(z), 0, 0)$
		\begin{align*}
		E(z) =&(\sigma-1) (w(z)-zw^{\prime}(z)) \left( X_{\geq 1}(z) - (z w^{\prime}(z))_{\geq 1} \frac{X(z)}{z w^{\prime}(z)} + \frac{X(z)}{w^{\prime}(z)} + \frac{e^u}{z} \frac{X(z)}{z w^{\prime}(z)} \right) \\ 
		+& (\sigma - 1) z w^{\prime} (z) \left( \left( (w(z) - z w^{\prime}(z))_{>0} \frac{X(z)}{z w^{\prime}(z)} \right)_{<0} - \left( (w(z) - zw^{\prime}(z))_{\leq 0} \frac{X(z)}{z w^{\prime}(z)} \right)_{\geq 0} \right) \\ 
		+& (\sigma - 1) \left( 2 \frac{e^u}{z} X(z) + v X(z) \right) + \left(2 \frac{e^u}{z} X(z) + v X(z)  \right)_{\geq 1} \\
		+& \left( (w(z)-zw^{\prime}(z)) \left( X_{\geq 1}(z) - (z w^{\prime}(z))_{\geq 1} \frac{X(z)}{z w^{\prime}(z)} + \frac{X(z)}{w^{\prime}(z)} + \frac{e^u}{z} \frac{X(z)}{z w^{\prime}(z)} \right) \right)_{\geq 1} \\
		+& \left( z w^{\prime} (z) \left( \left( (w(z) - z w^{\prime}(z))_{>0} \frac{X(z)}{z w^{\prime}(z)} \right)_{<0} - \left( (w(z) - zw^{\prime}(z))_{\leq 0} \frac{X(z)}{z w^{\prime}(z)} \right)_{\geq 0} \right) \right)_{\geq 1} \\
		+& e^u \left( (w(z) + z w^{\prime}(z )) \frac{X(z)}{z w^{\prime}(z)} \right)_1 + \frac{e^u}{z} \left( (w(z) - z w^{\prime}(z)) \frac{X(z)}{z w^{\prime}(z)} \right)_0 \\
		-& \left( (\sigma - 1) w(z) + w_{\geq 1} (z) - z + v + \frac{e^u}{z} \right) \left( (\sigma-1) X(z) + X_{\geq 1}(z) \right) \\
		+& \left( (\sigma - 1) z w^{\prime}(z) + (z w^{\prime}(z))_{\geq 1} - z - \frac{e^u}{z} \right) \left( (\sigma-1) \frac{w(z)}{z w^{\prime}(z)} X(z) + \left(\frac{w(z)}{z w^{\prime}(z)} X(z ) \right)_{\geq 1} \right).
		\end{align*}
		It is immediate to see that the terms with $(\sigma -1 )^2$, $v$, and $(\sigma-1)e^u$ cancel out. First, we simplify the rest of the terms with $e^u$, which equal
		\begin{align*}
		-\frac{e^u}{z} \left( \left(1-\frac{w(z)}{z w^{\prime}(z)} \right) X(z) \right)_0.
		\end{align*}
		Second, one can similarly see that the terms with $(\sigma-1)$ equal
		\begin{align}
		(\sigma - 1) z w^{\prime}(z) \left( \left(1-\frac{w(z)}{z w^{\prime}(z)} \right) X(z) \right)_0. \notag
		\end{align}
		Third, we split the remaining terms of $E(z)$ into two groups, the first one being
		\begin{align*}
		&z X_{\geq 1}(z) - z \left( \frac{w(z)}{z w^{\prime}(z)} X(z)  \right)_{\geq 1} + \left( \frac{w(z)}{w^{\prime}(z)} X(z) \right)_{\geq 1} - (z X(z))_{\geq 1} \\ &= -z  \left( \left(1-\frac{w(z)}{z w^{\prime}(z)} \right) X(z) \right)_0.
		\end{align*}
		Finally, we are left with
		\begin{align*}
		-& w_{\geq 1}(z) X_{\geq 1}(z) + (w(z) X_{\geq 1}(z))_{\geq 1} + (z w^{\prime}(z))_{\geq 1} \left( \frac{w(z)}{z w^{\prime}(z)} X(z) \right)_{\geq 1} \\
		-&\left( (zw^{\prime}(z))_{\geq 1} \frac{w(z)}{z w^{\prime}(z)} X(z ) \right)_{\geq 1} - (z w^{\prime} (z) X(z))_{\geq 1} + \left( ( zw^{\prime}(z))_{\geq 1} X(z) \right)_{\geq 1} \\
		+& \left( z w^{\prime} (z) \left( \left( (w(z) - z w^{\prime}(z))_{>0} \frac{X(z)}{z w^{\prime}(z)} \right)_{<0} - \left( (w(z) - zw^{\prime}(z))_{\leq 0} \frac{X(z)}{z w^{\prime}(z)} \right)_{\geq 0} \right) \right)_{\geq 1} \\ 
		=& (z w^{\prime}(z))_{\geq 1} \left( \left(1-\frac{w(z)}{z w^{\prime}(z)} \right) X(z) \right)_0.
		\end{align*}
		Putting everything together, 
		\begin{align}
		E(z) &= \left( (\sigma - 1) z w^{\prime}(z) + (z w^{\prime}(z))_{\geq 1} - z - \frac{e^u}{z} \right) \left( \left(1-\frac{w(z)}{z w^{\prime}(z)} \right) X(z) \right)_0 \\ &= z \lambda_\sigma^{\prime}(z) \left( \left(1-\frac{w(z)}{z w^{\prime}(z)} \right) X(z) \right)_0 \notag.
		\end{align}
		Let $\hat{X} = (0,1,0)$. In this case, it is immediate to see
		\begin{align}
		E(z) = -z \lambda_\sigma^{\prime}(z).
		\end{align}
		Finally, for $\hat{X} = (0,0,1)$, it is also a straightforward computation to check
		\begin{align}
		E(z) = 0,
		\end{align}
		concluding the proof.
\end{proof}

\subsection{Proof of Proposition~\ref{prop:canonical} and Corollary~\ref{cor:representation}} 
\label{sec:proooof}

Let $\hat{\lambda} = (\lambda(z), \bar{\lambda}(z)) \in M_0$ be such that $\lambda(z)$ has $n$ critical points in the interior of $D_{\infty}$ and $\bar{\lambda}(z)$ has $\bar{n}$ critical points in the interior of $D_0$. We take $\hat{\lambda}$ to be generic, i.e., none of the critical points is degenerate. 

The fact that the functionals $d\lambda(z_i)$, $d \bar\lambda(\bar{z}_i)$ and $d \lambda_\sigma (p)$ for $\sigma = \sigma(p)$ are generalized eigenvectors of $\cU$ simply follows from Lemma~\ref{lem:keylemma}. 
Indeed, let $z_i$ be one of the critical points of $\lambda(z)$, i.e. $\lambda'(z_i) =0$; substituting $\sigma=0$ and $z=z_i$ in~\eqref{key-for}, we get at once that
\begin{equation}
\langle d\lambda(z_i) , \cU \hat{X} \rangle = - \lambda(z_i) \langle d\lambda(z_i), \hat{X} \rangle
\end{equation}
for all $\hat{X}$, namely $d\lambda(z_i)$ is a generalized eigenvector corresponding to the eigenvalue $u_i = -\lambda(z_i)$. Similarly, setting $\sigma=1$ and $z=\bar{z}_i$, resp. $\sigma=\sigma(p)$ and $z=p$, we obtain the analogous statement for $d\bar\lambda(\bar{z}_i)$ and $\bar{u}_i$, resp. $(d\lambda_\sigma (p))|_{\sigma=\sigma(p)}$ and $u_p$.
By~\eqref{eqn:diffdiscrete} and~\eqref{eqn:diffcanonical}, we have
\begin{equation}
du_i = d \lambda(z_i), \qquad 
d\bar{u}_i = d \bar\lambda(\bar{z}_i), \qquad 
du_p = d\lambda_\sigma (p)|_{\sigma=\sigma(p)} .
\end{equation}
One can easily check that $du_i$ and $d\bar{u}_i$ are representable as~\eqref{eqn:duirep}, therefore they are eigenvectors. 

Let us now prove that this family of generalized eigenvectors is complete. For that, we will prove that the map 
\begin{align}
\Psi: T_{\hat{\lambda}}M &\longrightarrow \cH (S^1) \oplus \C^n \oplus \C^{\bar{n}} \\
\hat{X}  &\longmapsto ( \langle du_p, \hat{X} \rangle, \langle  du_i, \hat{X} \rangle, \langle d\bar{u}_j, \hat{X} \rangle )
\end{align}
defines an isomorphism of vector spaces.
Let us consider tangent vectors as pairs $\hat{X} = (X(z), \bar{X}(z)) \in  \cH(D_{\infty}) \oplus \frac{1}{z} \cH(D_0)$, and let 
\begin{equation} \label{eq:utt}
Y(p) = \langle du_p, \hat{X} \rangle,\qquad  
Y_i = \langle du_i, \hat{X} \rangle, \qquad 
\bar{Y}_i = \langle d\bar{u}_i, \hat{X} \rangle. 
\end{equation}
Explicitly, 
\begin{align}
Y(p) = \frac{\lambda^{\prime}(p)}{\lambda^{\prime}(p) + \bar{\lambda}^{\prime}(p)} \bar{X}(p) -  \frac{\bar{\lambda}^{\prime}(p)}{\lambda^{\prime}(p) + \bar{\lambda}^{\prime}(p)} X(p), \qquad
Y_i = -X(z_i), \qquad \bar{Y}_i = \bar{X}(\bar{z}_i).
\end{align}
It is enough to observe that the inverse $\Psi^{-1}$  is given by
\begin{align}
\bar{X}(p) &= \bar{\lambda}^{\prime}(p) \left[ \mu_{\hat{Y}}(p) + \left( \frac{\lambda^{\prime}(p) + \bar{\lambda}^{\prime}(p)}{\lambda^{\prime}(p)\bar{\lambda}^{\prime}(p)} Y(p) \right)_{\geq 1}  \right], \label{eqn:inversepsi1} \\
X(p) &= -\lambda^{\prime}(p) \left[ -\mu_{\hat{Y}}(p) + \left( \frac{\lambda^{\prime}(p) + \bar{\lambda}^{\prime}(p)}{\lambda^{\prime}(p)\bar{\lambda}^{\prime}(p)} Y(p) \right)_{\leq 0}  \right], \label{eqn:inversepsi2}
\end{align}
where
\begin{align}
\mu_{\hat{Y}}(p) =  \sum_{i=1}^{n} \frac{Y_i}{z_i \lambda''(z_i)} \frac{p}{z_i - p} - \sum_{i=1}^{\bar{n}} \frac{\bar{Y}_i}{\bar{z}_i \bar\lambda''(\bar{z}_i)} \frac{p}{\bar{z}_i - p}. \label{eqn:muY}
\end{align}

Let $\hat Y = (Y(p), Y_i, \bar{Y}_i) \in  \cH (S^1) \oplus \C^n \oplus \C^{\bar{n}} $ and $\hat{X} = \Psi^{-1} \hat{Y}$. 
Corollary~\ref{cor:representation} follows from observing that
\begin{equation}
\Psi \cU \hat{X} =  ( \langle du_p, \cU\hat{X} \rangle, \langle  du_i, \cU\hat{X} \rangle, \langle d\bar{u}_j, \cU\hat{X} \rangle ) =
(u_p Y(p),  u_i Y_i,  \bar{u}_j \bar{Y}_j  ).
\end{equation}

To conclude, we notice that the (generalized) eigenvalues of $\cU$ coincide with those of $U$, the eigenvectors being related by the isomorphism $\Psi$.
Notice that, since $\lambda_\sigma'(p)=0$ for $\sigma = \sigma(p)$, 
we have 
\begin{equation}
\frac{d u_p}{dp} = \sigma'(p) w(p).  \label{eqn:dupdp}
\end{equation}
Therefore, because of axioms~\ref{item:T3} and~\ref{item:T5}, $\frac{d u_p}{dp}$ is non-vanishing on $S^1$. 
This implies that the generalized eigenspaces are only those given in the proposition, see the following remarks for further details. 
\qed

\begin{remark}
Consider the operator $U$ on $\cH(S^1) \oplus \C^n \oplus \C^{\bar{n}}$ given in Corollary~\ref{cor:representation}, namely
\begin{equation}
U(X(z), X_i, \bar{X}_j ) = ( u_z X(z) , u_i X_i , \bar{u}_j \bar{X}_j ) 
\end{equation}
for $u_z \in \cH(S^1)$.
Clearly, the set of generalized and standard eigenvalues corresponds to the set of $u_p$, $u_i$ and $\bar{u}_j$, namely

\begin{lemma}
The spectrum of $U$ is given by 
$\{ u_p, u_i, \bar{u}_j \}_{p \in S^1, i=1, \dots, n, j = 1, \dots, \bar{n}}$.
\end{lemma}

\begin{proof}
From the diagonal form of $U$, it is immediately clear that the standard eigenvalues are $\{ u_i \}_{i = 1, \dots, n}$ and $ \{ \bar{u}_j \}_{j = 1, \dots, \bar{n}}$ with eigenvectors $(0,e_i, 0)$ and $(0, 0, e_j)$, respectively, where $e_i$ denotes the canonical basis vector which is $1$ at the $i$-th entry and $0$ everywhere else. 

In order to find its generalized eigenvalues, we look for $\lambda \in \C$, $0 \not= \xi \in T^*_{\hat{\lambda}} M$ such that
\begin{align}
\langle \xi, \left( (u_z - \lambda) X(z), (u_i - \lambda) X_i, (\bar{u}_j - \lambda) \bar{X}_j \right) \rangle = 0, \qquad \forall \hat{X} \in \cH(S^1) \oplus \C^n \oplus \C^{\bar{n}}.
\end{align}
Consider the decomposition $\xi = (\xi_z, \xi_i, \bar{\xi}_j)$ given by 
\begin{align}
\langle \xi, \hat{X} \rangle = \langle \xi_z, X(z) \rangle + \sum_{i=1}^n X_i \langle \xi_i, e_i \rangle + \sum_{j=1}^{\bar{n}} \bar{X}_j \langle \bar{\xi}_j, e_j \rangle.
\end{align}
For $p \in S^1$, one can check that $u_p$ is a generalized eigenvalue with generalized eigenvector $(\textrm{ev}_{p}, 0, 0)$, where the functional $\textrm{ev}_{p}$ is defined by
\begin{align}
\langle \mathrm{ev}_{p} , X(z) \rangle  = X(p). \label{eqn:defnev}
\end{align}

Finally, let $\lambda \not= u_i, \bar{u}_j, u_p$ for any $i, j, p$. Since $\lambda \not= u_i, \bar{u}_j$, then we have $\xi_i = \bar{\xi}_j = 0$ for all $i, j$, so we are left with
\begin{align}
\langle \xi_z, (u_z - \lambda) X(z) \rangle =0, \qquad \forall X \in \cH(S^1). \label{eqn:equationduplambda}
\end{align}
Since $\lambda \not= u_p$ for any $p \in S^1$, then multiplication by $(u_z - \lambda)$ is an invertible operator in $\cH(S^1)$, so $\xi_z = 0$, hence $\lambda$ is not an eigenvalue.
\end{proof}

Let us now compute the dimension of the (generalized) eigenspaces. Since $\frac{du_z}{dz}$ does not vanish on $S^1$, we have the following
\begin{lemma}
Suppose exactly $s+ k+ \ell$ generalized eigenvalues coincide, namely
\begin{align}
u_{p_1} = \dots = u_{p_s} = u_{i_1} = \dots = u_{i_k} = \bar{u}_{j_1} = \dots = \bar{u}_{j_\ell}. \label{eqn:lambdaiseigen}
\end{align}
Then the corresponding eigenspace is $s + k +\ell$ dimensional. 
\end{lemma}

\begin{proof}
Let $\lambda$ denote~\eqref{eqn:lambdaiseigen}. Then the generalized eigenspace of $\lambda$ splits into two subspaces, the $(k+\ell)$-dimensional subspace corresponding to the standard eigenvectors $\{(0, e_{i_r}, 0),  (0,e_{j_l}, 0) \}_{r = 1, \dots, k}^{l = 1, \dots, \ell}$ mentioned before, and the subspace given by $\xi = (\xi_z, 0, 0)$ with $\xi_z$ satisfying equation~\eqref{eqn:equationduplambda}.
Let us compute the latter for $s \geq 1$. By~\eqref{eqn:dupdp} and axioms~\ref{item:T3} and~\ref{item:T5}, the function $\frac{du_z}{dz}$ does not vanish on $S^1$, so $u_z - \lambda$ does not have double zeros on $S^1$, i.e.,
\begin{align}
u_z - \lambda = (z-p_1)\dots (z-p_s) g(z),  \label{eqn:lambdadoesnotdouble}
\end{align}
where $g(z)$ is a non-vanishing holomorphic function on $S^1$. Therefore, since multiplication by $g(z)$ is invertible on $\cH(S^1)$, equation~\eqref{eqn:equationduplambda} becomes
\begin{align}
\langle \xi_z, (z-p_1) \dots (z-p_s)  X(z) \rangle =0, \qquad \forall X \in \cH(S^1), \label{eqn:z-pz-q1}
\end{align}
or, equivalently, $\xi_z$ vanishes on the subspace of $\cH(S^1)$ given by functions with zeros at the distinct points $p_1, \dots, p_s$. It is clear that the functionals $ \textrm{ev}_{p_1}, \dots, \mathrm{ev}_{p_s}$ defined in~\eqref{eqn:defnev} are linearly independent and solve~\eqref{eqn:z-pz-q1}. Let us show that they span the whole space of solutions of~\eqref{eqn:z-pz-q1}. 

For that, we need the following decomposition formula: for any $X \in \cH(S^1)$, $s \geq 1$, we can write
	\begin{align}
	X(z) = X(p_1) + (z-p_1) Y_1 + (z-p_1)(z-p_2) Y_2 + \dots +(z-p_1) \dots (z-p_s) Y_s \label{eqn:Xdecomp}
	\end{align}
	where $Y_{<s} \in \C$, $Y_s \in \cH(S^1)$. This statement can be easily proved by induction. For $s=1$, it is clear by taking
	\begin{align}
	Y_1 (z) = \frac{1}{z-p_1} \left( X(z) - X(p_1) \right).
	\end{align}
	Assuming it holds for $s-1 \geq 1$, we write 
	\begin{align}
	X(z) &= X(p_1) + (z-p_1) Y_1 + (z-p_1)(z-p_2) Y_2 + \dots +(z-p_1) \dots (z-p_{s-1}) Y_{s-1}(z) \\ &= X(p_1) + (z-p_1) Y_1 + (z-p_1)(z-p_2) Y_2 + \dots +(z-p_1) \dots (z-p_{s-1}) Y_{s-1}(p_s) \notag \\ &+ (z-p_1) \dots (z-p_{s}) \frac{Y_{s-1} (z) - Y_{s-1}(p_s) }{z-p_s},  \notag
	\end{align}
	where we have split
	\begin{align}
	Y_{s-1} (z) = Y_{s-1} (p_s) + (z-p_s) \frac{Y_{s-1} (z) - Y_{s-1}(p_s) }{z-p_s}.
	\end{align}

Applying~\eqref{eqn:Xdecomp}, we write
\begin{align}
\langle \xi_z, X(z) \rangle &= X(p_1) \langle \xi_z, 1 \rangle  + Y_1 \langle \xi_z, z-p_1 \rangle + \dots + Y_{s-1} \langle \xi_z, (z-p_1) \dots (z-p_{s-1}) \rangle \\ &+  \langle \xi_z, (z-p_1) \dots (z-p_s) Y_s(z) \rangle \notag.
\end{align}
The last summand vanishes because $\xi_z$ satisfies equation~\eqref{eqn:z-pz-q1}. Therefore, $\xi_z$ is completely determined by the numbers
\begin{align}
\langle \xi_z, 1 \rangle, \langle \xi_z, z \rangle, \dots, \langle \xi_z, z^{s-1} \rangle,
\end{align}
so the space of solutions of~\eqref{eqn:z-pz-q1} is at most $s$-dimensional, hence it must be the span of $\textrm{ev}_{p_1}, \dots, \mathrm{ev}_{p_s}$.
\end{proof}
\end{remark}

\begin{remark}
Notice that relaxing the axioms~\ref{item:T3} and~\ref{item:T5} in the definition of $M_0$ would imply, by relation~\eqref{eqn:dupdp}, dropping  the non-vanishing assumption of the derivative of $u_z$ on $S^1$.
In such case the function $u_z - \lambda$ might have higher order zeros, i.e.,
\begin{align}
u_z - \lambda = (z-p_1)^{N_1} (z-p_2)^{N_2} \dots (z-p_s)^{N_s} g(z)
\end{align}
with $N_i \geq 1$. Then the subspace determined by equation~\eqref{eqn:equationduplambda} is $(N_1 + \dots + N_s)$-dimensional, generated by the functionals
\begin{align}
\langle \mathrm{ev}_{p_i}^{(m)}, X(z) \rangle = \frac{d^{m}}{dz^{m}}\bigg|_{z = p_i} X(z) , \qquad m = 0, \dots, N_i - 1.
\end{align}
\end{remark}

\subsection{Metric in canonical coordinates}

Thanks to the explicit expression for $\Psi^{-1}$, we can derive the following diagonal form of the metric in canonical coordinates:
\begin{proposition} \label{coro:metriccanonical}
The metric $\eta$ in the representation given by $\Psi$, i.e. $\tilde{\eta}(\hat{X}, \hat{Y}) := \eta( \Psi^{-1}\hat{X}, \Psi^{-1}\hat{Y}) $ has the diagonal form
\begin{align}
\tilde{\eta}(\hat{X}, \hat{Y}) &= - \frac{1}{2 \pi \ii} \oint_{|z| = 1} \frac{w'(z)}{\lambda'(z) \bar{\lambda}'(z)} X(z) Y(z) \frac{dz}{z^2} - \label{eqn:tildeetaincanonical}\\
&-\sum_{i=1}^n \frac{1}{z_i^2 \lambda''(z_i)} X_i Y_i + \sum_{j = 1}^{\bar{n}} \frac{1}{\bar{z}_j^2 \bar{\lambda}''(\bar{z}_j)} \bar{X}_j \bar{Y}_j \notag
\end{align}
for $\hat{X}, \hat{Y} \in \cH (S^1) \oplus \C^n \oplus \C^{\bar{n}}$.
\end{proposition}

\begin{proof}
First, using~\eqref{eqn:metricdefn},~\eqref{eqn:inversepsi1}--\eqref{eqn:inversepsi2} and~\eqref{eqn:pairstows}, we compute 
\begin{align}
\tilde{\eta}( (0, e_i, 0) , (0, e_j, 0)) &= \frac{1}{z_i \lambda''(z_i) z_j \lambda''(z_j)} \left( \frac{1}{2 \pi \ii} \oint_{|z| = 1} w'(z) \frac{1}{z_i-z} \frac{1}{z_j - z} dz \right. \label{eqn:tildeXiYj} + \\ &+ e^{-u} \left. \left[ \left(  \frac{z \bar{\lambda}'(z)}{z_i-z} \right)_0 \left(  \frac{z \bar{\lambda}'(z)}{z_j-z} \right)_{-1} + \left(  \frac{z \bar{\lambda}'(z)}{z_i-z} \right)_{-1} \left(  \frac{z \bar{\lambda}'(z)}{z_j-z} \right)_{0}  \right] \right) \notag.
\end{align}
Notice that $z(z_i-z) = (z(z_i-z))_{\geq 1}$ and $\bar{\lambda}'(z) = - e^u z^{-2} + (\bar{\lambda}'(z))_{\geq 0}$, therefore
\begin{align}
&\left(  \frac{z \bar{\lambda}'(z)}{z_i-z} \right)_0 = -e^u \left(  \frac{z}{z_i-z} \right)_2 = - \frac{e^u}{z_i^2}, \\ 
&\left(  \frac{z \bar{\lambda}'(z)}{z_i-z} \right)_{-1} = -e^u \left(  \frac{z}{z_i-z} \right)_1 = - \frac{e^u}{z_i}. \label{eqn:residues01}
\end{align}
On the other hand, we can split the integral of~\eqref{eqn:tildeXiYj} by decomposing $w'(z) = \bar{\lambda}'(z) + \lambda'(z)$. The first summand equals
\begin{align}
&\frac{1}{2 \pi \ii} \oint_{|z| = 1} \bar{\lambda}'(z) \frac{1}{z_i-z} \frac{1}{z_j - z} dz = \res_{z=0} \bar{\lambda}'(z) \frac{1}{z_i-z} \frac{1}{z_j - z} \label{eqn:residuelambdabar} \\ &= -e^{u} \res_{z=0} \frac{1}{z^2} \frac{1}{z_i-z} \frac{1}{z_j - z} 
= - \frac{e^u}{z_i z_j} \left( \frac{1}{z_i} + \frac{1}{z_j} \right). \notag
\end{align}
Plugging~\eqref{eqn:residues01} and~\eqref{eqn:residuelambdabar} in~\eqref{eqn:tildeXiYj} yields
\begin{align}
&\tilde{\eta}( (0, e_i, 0) , (0, e_j, 0)) = \frac{1}{z_i \lambda''(z_i) z_j \lambda''(z_j)}  \frac{1}{2 \pi \ii} \oint_{|z| = 1} \lambda'(z) \frac{1}{z_i-z} \frac{1}{z_j - z} dz = \\ 
&= \frac{1}{z_i \lambda''(z_i) z_j \lambda''(z_j)} \begin{cases}
- \res_{z=z_i} \lambda'(z) \frac{1}{z_i-z} \frac{1}{z_j - z} - \res_{z=z_j} \lambda'(z) \frac{1}{z_i-z} \frac{1}{z_j - z}, & i \not=j \\
- \res_{z=z_i} \lambda'(z) \frac{1}{(z_i-z)^2}, & i=j
\end{cases} \notag \\
&= \frac{1}{(z_i \lambda''(z_i))^2} \begin{cases}
0, & i \not=j \\
- \lambda''(z_i), & i=j \notag
\end{cases} \\ &= -\frac{1}{z_i^2 \lambda''(z_i)} \delta_{ij} \notag.
\end{align}
Analogously, one obtains
\begin{align}
&\tilde{\eta}( (0, e_i, 0) , (0, 0, e_j)) = \tilde{\eta}( (0, 0, e_i) , (0, e_j, 0)) = 0, \\ 
&\tilde{\eta}( (0, 0, e_i) , (0, 0, e_j)) =  \frac{1}{\bar{z}_i^2 \bar{\lambda}''(\bar{z}_i)} \delta_{ij}.
\end{align}

As before, we use formulas~\eqref{eqn:metricdefn},~\eqref{eqn:inversepsi1}--\eqref{eqn:inversepsi2} and~\eqref{eqn:pairstows} to compute
\begin{align}
\label{eqn:tildeetaxy} 
&\tilde{\eta} ((X(z), 0, 0), (Y(z), 0, 0)) =\\ 
&=\frac{1}{2 \pi \ii} \oint_{|z| = 1} \left[ \lambda'(z)^2 \left( \frac{w' X}{\lambda' \bar{\lambda}'} \right)_{\leq 0} \hspace{-4mm} (z) \left( \frac{w' Y}{\lambda' \bar{\lambda}'}  \right)_{\leq 0} \hspace{-4mm} (z) + \bar{\lambda}'(z)^2 \left( \frac{w' X}{\lambda' \bar{\lambda}'} \right)_{\geq 1} \hspace{-4mm} (z) \left( \frac{w' Y}{\lambda' \bar{\lambda}'}  \right)_{\geq 1} \hspace{-4mm} (z) \right. \notag\\
&\qquad \quad - \left. \lambda'(z) \bar{\lambda}'(z) \left( \left( \frac{w' X}{\lambda' \bar{\lambda}'} \right)_{\leq 0} \hspace{-4mm} (z) \left( \frac{w' Y}{\lambda' \bar{\lambda}'}  \right)_{\geq 1} \hspace{-4mm} (z) + \left( \frac{w' X}{\lambda' \bar{\lambda}'} \right)_{\geq 1} \hspace{-4mm} (z) \left( \frac{w' Y}{\lambda' \bar{\lambda}'}  \right)_{\leq 0} \hspace{-4mm} (z) \right)  \right] \frac{dz}{z^2 w'(z)}  \notag \\
&+e^{-u} \left[ \left(\bar{\lambda}' \left( \frac{w'X}{\lambda' \bar{\lambda}'}  \right)_{\geq 1} \hspace{-4mm} (z)  \right)_{\hspace{-1mm}0} \left(\bar{\lambda}' \left( \frac{w'Y}{\lambda' \bar{\lambda}'}  \right)_{\geq 1} \hspace{-4mm} (z)  \right)_{\hspace{-1mm}-1} \notag 
\hspace{-4mm}+ \left(\bar{\lambda}' \left( \frac{w'X}{\lambda' \bar{\lambda}'}  \right)_{\geq 1} \hspace{-4mm} (z)  \right)_{\hspace{-1mm}-1} \left(\bar{\lambda}' \left( \frac{w'Y}{\lambda' \bar{\lambda}'}  \right)_{\geq 1} \hspace{-4mm} (z)  \right)_{\hspace{-1mm}0}  \right]. \notag
\end{align}
Since $\bar{\lambda}'(z) = -e^u z^{-2} + (\bar{\lambda}'(z))_{\geq 0}$, we have
\begin{align}
&\left(\bar{\lambda}'(z) \left( \frac{w'X}{\lambda' \bar{\lambda}'}  \right)_{\geq 1} \hspace{-4mm} (z)  \right)_0 = - e^u  \left( \frac{w'X}{\lambda' \bar{\lambda}'}  \right)_{2}, \\
&\left(\bar{\lambda}'(z) \left( \frac{w'X}{\lambda' \bar{\lambda}'}  \right)_{\geq 1} \hspace{-4mm} (z)  \right)_{-1} = - e^u  \left( \frac{w'X}{\lambda' \bar{\lambda}'}  \right)_{1}, \label{eqn:lambdaX0}
\end{align}
and the same for $Y(z)$.
Let us consider the integral in~\eqref{eqn:tildeetaxy}. We can rewrite the first summand as
\begin{align}
\label{eqn:metriccan1}
&\frac{1}{2 \pi \ii} \oint_{|z| = 1} \lambda'(z)(w'(z) - \bar{\lambda}'(z)) \left( \frac{w' X}{\lambda' \bar{\lambda}'} \right)_{\leq 0} \hspace{-4mm} (z) \left( \frac{w' Y}{\lambda' \bar{\lambda}'}  \right)_{\leq 0} \hspace{-4mm} (z) \frac{dz}{z^2 w'(z)} = \\ 
&= \frac{1}{2 \pi \ii} \oint_{|z| = 1} \lambda'(z) \left( \frac{w' X}{\lambda' \bar{\lambda}'} \right)_{\leq 0} \hspace{-4mm} (z) \left( \frac{w' Y}{\lambda' \bar{\lambda}'}  \right)_{\leq 0} \hspace{-4mm} (z) \frac{dz}{z^2}- \notag\\
&- \frac{1}{2 \pi \ii} \oint_{|z| = 1} \lambda'(z)\bar{\lambda}'(z) \left( \frac{w' X}{\lambda' \bar{\lambda}'} \right)_{\leq 0} \hspace{-4mm} (z) \left( \frac{w' Y}{\lambda' \bar{\lambda}'}  \right)_{\leq 0} \hspace{-4mm} (z) \frac{dz}{z^2 w'(z)} \notag.
\end{align}
Notice the first summand on the right-hand side vanishes because $\lambda'(z) = (\lambda'(z))_{\leq 0}$, so the integrand has no residue. Similarly, 
\begin{align}
&\frac{1}{2 \pi \ii} \oint_{|z| = 1} \bar{\lambda}'(z)^2 \left( \frac{w' X}{\lambda' \bar{\lambda}'} \right)_{\geq 1} \hspace{-4mm} (z) \left( \frac{w' Y}{\lambda' \bar{\lambda}'}  \right)_{\geq 1} \hspace{-4mm} (z) \frac{dz}{z^2 w'(z)} \label{eqn:metriccan2} \\
&= \frac{1}{2 \pi \ii} \oint_{|z| = 1} \bar{\lambda}'(z) \left( \frac{w' X}{\lambda' \bar{\lambda}'} \right)_{\geq 1} \hspace{-4mm} (z) \left( \frac{w' Y}{\lambda' \bar{\lambda}'}  \right)_{\geq 1} \hspace{-4mm} (z) \frac{dz}{z^2 } \notag -\\
&- \frac{1}{2 \pi \ii} \oint_{|z| = 1} \lambda'(z) \bar{\lambda}'(z) \left( \frac{w' X}{\lambda' \bar{\lambda}'} \right)_{\geq 1} \hspace{-4mm} (z) \left( \frac{w' Y}{\lambda' \bar{\lambda}'}  \right)_{\geq 1} \hspace{-4mm} (z) \frac{dz}{z^2 w'(z)} \notag.
\end{align}
The first summand on the right-hand side equals
\begin{align}
&-\frac{1}{2 \pi \ii} \oint_{|z| = 1} \frac{e^u}{z^2} \left( \frac{w' X}{\lambda' \bar{\lambda}'} \right)_{\geq 1} \hspace{-4mm} (z) \left( \frac{w' Y}{\lambda' \bar{\lambda}'}  \right)_{\geq 1} \hspace{-4mm} (z) \frac{dz}{z^2 }=\\
 &= -e^u \left[ \left( \frac{w' X}{\lambda' \bar{\lambda}'} \right)_{2} \left( \frac{w' Y}{\lambda' \bar{\lambda}'}  \right)_{1} +  \left( \frac{w' X}{\lambda' \bar{\lambda}'} \right)_{1} \left( \frac{w' Y}{\lambda' \bar{\lambda}'}  \right)_{2} \right], \notag
\end{align}
which cancels out with the third line of~\eqref{eqn:tildeetaxy} by~\eqref{eqn:lambdaX0}.
By replacing~\eqref{eqn:metriccan1} and~\eqref{eqn:metriccan2} in~\eqref{eqn:tildeetaxy} and noting that $A(z) B(z) = A_{\leq 0} (z) B_{\leq 0}(z) + A_{\leq 0} (z) B_{\geq 1}(z) + A_{\geq 1} (z) B_{\leq 0}(z) + A_{\geq 1} (z) B_{\geq 1}(z)$ for any $A, B \in \cH(S^1)$, we have
\begin{align}
\tilde{\eta} ((X(z), 0, 0), (Y(z), 0, 0))  = -\frac{1}{2 \pi \ii} \oint_{|z| = 1} \frac{w'(z)}{\lambda'(z) \bar{\lambda}'(z)} X(z) Y(z) \frac{dz}{z^2}. \label{eqn:tildemetric2}
\end{align}
Finally, let us compute 
\begin{align}
&\tilde{\eta} ((X(z), 0, 0), (0, e_i, 0)) = \frac{1}{z_i \lambda''(z_i)} \left( \frac{1}{2 \pi \ii} \oint_{|z| = 1} \bar{\lambda}'(z) \left( \frac{w' X}{\lambda' \bar{\lambda}'} \right)_{\geq 1} \hspace{-4mm} (z) \frac{1}{z_i - z} \frac{dz}{z} \right. \label{eqn:X(z)withei}
\\ &- \frac{1}{2 \pi \ii} \oint_{|z| = 1} \lambda'(z) \left( \frac{w' X}{\lambda' \bar{\lambda}'} \right)_{\leq 0} \hspace{-4mm} (z) \frac{1}{z_i - z} \frac{dz}{z}
+ \left. \frac{e^u}{z_i} \left[ \frac{1}{z_i} \left( \frac{w' X}{\lambda' \bar{\lambda}'} \right)_{1} + \left( \frac{w' X}{\lambda' \bar{\lambda}'} \right)_{2} \right] \right), \notag
\end{align}
where we have used~\eqref{eqn:residues01} and~\eqref{eqn:lambdaX0}. Note $|z_i| > 1$, so 
\begin{align}
\frac{1}{z_i - z} = \frac{1}{z_i} \sum_{k=0}^{\infty} \left( \frac{z}{z_i} \right)^k,
\end{align}
and the first integral of~\eqref{eqn:X(z)withei} becomes
\begin{align}
- \frac{1}{2 \pi \ii} \oint_{|z| = 1} \frac{e^u}{z^2} \left( \frac{w' X}{\lambda' \bar{\lambda}'} \right)_{\geq 1} \hspace{-4mm} (z) \frac{1}{z_i - z} \frac{dz}{z} = - e^u \left[ \frac{1}{z_i^2} \left( \frac{w' X}{\lambda' \bar{\lambda}'} \right)_{1} + \frac{1}{z_i} \left( \frac{w' X}{\lambda' \bar{\lambda}'} \right)_{2} \right],
\end{align}
which cancels with the third summand of~\eqref{eqn:X(z)withei}. The remaining term
\begin{align}
- \frac{1}{2 \pi \ii} \oint_{|z| = 1} \lambda'(z) \left( \frac{w' X}{\lambda' \bar{\lambda}'} \right)_{\leq 0} \hspace{-4mm} (z) \frac{1}{z_i - z} \frac{dz}{z} = \res_{z = z_i} \lambda'(z) \left( \frac{w' X}{\lambda' \bar{\lambda}'} \right)_{\leq 0} \hspace{-4mm} (z) \frac{1}{z_i - z} \frac{1}{z}
\end{align}
vanishes because $\lambda'(z_i) = 0$, so $\tilde{\eta} ((X(z), 0, 0), (0, e_i, 0))=0$. Analogously, one obtains
\begin{align}
\tilde{\eta} ((X(z), 0, 0), &(0, 0, e_j)) = \tilde{\eta} ((0, e_i, 0), (Y(z),  0, 0))=\\ &= \tilde{\eta} ((0, 0, e_i), (Y(z),  0, 0)) = 0, \notag
\end{align}
concluding the proof.
\end{proof}

\section{Dubrovin equation} 
\label{sec:Dubrovinequation}

It is well known that the geometric structure of a Dubrovin--Frobenius manifold is (almost) completely encoded in the flatness of the so-called deformed flat connection $\widetilde{\nabla}$, which is an extension to $M_0 \times  \C^*$ of the Levi-Civita connection of the metric $\eta$ obtained by deforming it using the associative product on the tangent bundle. 
In our case, if $\nabla$ denotes the Levi-Civita connection of the metric $\eta$, then the deformed flat connection $\widetilde{\nabla}$ on $M_0 \times \mathbb{C}^{*}$ is defined by~\cite{CM15}
\begin{align}
\widetilde{\nabla}_{\hat{X}} \hat{Y} &=  \nabla_{\hat{X}} \hat{Y} + \zeta \hat{X} \cdot \hat{Y}, \\ 
\widetilde{\nabla}_{\frac{d}{d \zeta} } \hat{X} &= \partial_\zeta \hat{X} + \mathcal{U} (\hat{X}) - \frac{1}{\zeta} \mathcal{V} (\hat{X}), \label{deformedflatangent} \\
\widetilde{\nabla}_{\hat{X}} \frac{d}{d \zeta} &= \widetilde{\nabla}_{\frac{d}{d \zeta}} \frac{d}{d \zeta} = 0,
\end{align}
for $\hat{X}, \hat{Y}\in T_{\hat{\lambda}} M$, where the operators $\mathcal{U}$ and $\mathcal{V}$ are given by~\eqref{eqn:defnU} and~\eqref{eqn:defnV}, respectively. 

In Dubrovin--Frobenius manifold theory, one is interested in looking for differentials $dy \in T^*_{\hat{\lambda}} M$ that are covariantly constant w.r.t. the deformed flat connection $\widetilde{\nabla}$. A basis of solutions adapted to $\zeta \sim 0$ provides a family of so-called deformed flat coordinates, the coefficients of which define the Hamiltonian densities of the principal hierarchy associated with the Dubrovin--Frobenius manifold. See~\cite{DZ01} for the general construction and~\cite{CM15} for the derivation of the  principal hierarchy of $M_0$. 

In this paper we focus on the Dubrovin equation, namely the flatness equation in the $\frac{d}{d\zeta}$ direction, corresponding to~\eqref{deformedflatangent} in the definition of the deformed flat connection. 
The covariant derivative w.r.t. $\frac{d}{d\zeta}$  
on an element of the cotangent space $\alpha \in T^*_{\hat{\lambda}}M$ depending on the deformation parameter $\zeta$ is given by
\begin{equation}
\widetilde{\nabla}_{\frac{d}{d \zeta}} (\alpha) = \partial_\zeta \alpha - \cU^* \alpha + \frac1\zeta \cV^*\alpha,
\end{equation}
where $\cU^*$ and $\cV^*$ denote the trasposes of $\cU$ and $\cV$. In other words, the cotangent vector $\widetilde{\nabla}_{\frac{d}{d \zeta}} (\alpha) $ is defined by 
\begin{align}
\langle\widetilde{\nabla}_{\frac{d}{d \zeta}} (\alpha), \hat{X} \rangle 
= \partial_\zeta \langle \alpha, \hat{X} \rangle - \langle \alpha, \left(\mathcal{U} - \frac{1}{\zeta} \mathcal{V} \right) \hat{X} \rangle, 
\end{align}
for all $\hat{X} \in T_{\hat{\lambda}}M$.

The Dubrovin equation $\widetilde{\nabla}_{\frac{d}{d \zeta}} (\alpha) = 0$ is therefore given by
\begin{equation} \label{eqn:generalizedDubrovin}
\partial_\zeta \langle \alpha, \hat{X} \rangle = \langle \alpha, \left(\mathcal{U} - \frac{1}{\zeta} \mathcal{V} \right) \hat{X} \rangle ,
\end{equation}
for all $\hat{X} \in T_{\hat{\lambda}}M$.
We look for deformed flat functionals $y(\zeta): M_0 \times \mathbb{C}^* \rightarrow \mathbb{C}$, namely those whose differential $dy(\zeta) \in T^*_{\hat{\lambda}}M$ is covariantly constant w.r.t. $\widetilde{\nabla}$. In particular, they are solutions of the Dubrovin equation~\eqref{eqn:generalizedDubrovin}. 

As expected, if the cotangent vector $\alpha$ is representable, $\alpha = \eta_* \hat{Z}$ for $\hat{Z} \in T_{\hat{\lambda}} M$, then~\eqref{eqn:generalizedDubrovin} is written as
\begin{equation}
\eta( \partial_\zeta \hat{Z} , \hat{X} ) = \eta( \hat{Z} ,  \left(\mathcal{U} - \frac{1}{\zeta} \mathcal{V} \right) \hat{X} ) ,
\end{equation}
which implies,  since $\mathcal{U}$ is symmetric and $\mathcal{V}$ antisymmetric with respect to the metric $\eta$, the usual form of Dubrovin equation
\begin{align}
\partial_\zeta \hat{Z} = \left( \mathcal{U} + \frac{1}{\zeta} \mathcal{V} \right) \hat{Z}, 
\label{eqn:repDub}
\end{align}
cf.~\cite[equation 20b]{CM15}.

\section{Formal solutions} 
\label{sec:formalsolutions}

Let us solve equation~\eqref{eqn:generalizedDubrovin} perturbatively at $\infty$. Recall that in the finite-dimensional case, the Dubrovin equation has a fundamental formal solution of the form (see~\cite{Dub99}) 
\begin{align}
\Xi (\zeta) = \Psi^{-1} R(\zeta) e^{U \zeta}, \quad R(\zeta)= \mathrm{Id} + R_1 \zeta^{-1} + R_2 \zeta^{-2} + \dots,
\end{align}
where $\Psi$ denotes the change of coordinates matrix from flat to normalized canonical. In the infinite-dimensional case, there is no natural analogue of the fundamental matrix $\Xi$, but nonetheless we can write functionals that generalize its columns, given by
\begin{align}
\xi_j (\zeta) = e^{\zeta u_j } \left( v_j^0 + v_j^1 \zeta^{-1} + \dots \right),
\end{align}
where $u_j$ is the $j$-th canonical coordinate and $v_j^k$ are constant column vectors. 
\begin{proposition} The following statements hold at any point $\hat{\lambda} \in M_0$: \label{prop:formal}
	\begin{enumerate}
		\item For each discrete canonical coordinate $u_i$, there exists a unique representable formal solution of the Dubrovin equation~\eqref{eqn:generalizedDubrovin} of the form \label{item:propformal1}
		\begin{align}
		\xi_i^{\textrm{formal}}(\zeta) = e^{\zeta u_i} \sum_{k=0}^\infty r_i^k \zeta^{-k}, \qquad r_i^0 = du_i, \ r_i^k \in (T^*_{\hat{\lambda}} M)^{\textrm{rep}}.
		\end{align}
		\item \label{item:propformal2} For each discrete canonical coordinate $\bar{u}_i$, there exists a unique representable formal solution of the Dubrovin equation~\eqref{eqn:generalizedDubrovin} of the form
		\begin{align}
		\bar{\xi}_i^{\textrm{formal}}(\zeta) = e^{\zeta \bar{u}_i} \sum_{k=0}^\infty \bar{r}_i^k \zeta^{-k}, \qquad \bar{r}_i^0 = d\bar{u}_i, \ \bar{r}_i^k \in (T^*_{\hat{\lambda}} M)^{\textrm{rep}}.
		\end{align}
		\item \label{item:propformal3} For each $p \in S^1$, the Dubrovin equation~\eqref{eqn:generalizedDubrovin} admits formal solutions of the form
		\begin{align}
		\xi_p^{\textrm{formal}}(\zeta) = e^{\zeta u_p} \sum_{k=0}^\infty r_p^k \zeta^{-k}, \qquad r_p^0 = du_p, \ r_p^k \in T_{\hat{\lambda}}^*M. \label{eqn:Ansatzformal}
		\end{align}
		These solutions are given by the functionals
		\begin{align}
		r_p^k &= \sum_{n = 0}^{k} a_p^n A^*_p (k-1-\cV)^* A_p^* (k-2-\cV)^* \dots A_p^* (n-\cV)^* du_p, \qquad a_p^0 = 1, \label{eqn:functionalsformal}
		\end{align}
		where $A_p$ is the left-inverse of $u_p - \cU$ with $A_p(0,1,0) = 0$, and depend on the choice of complex constants $a_p^n \in \C$ for $n \in \Z_{\geq 1}$.
	\end{enumerate}
\end{proposition}
\begin{proof}
	To prove items~\ref{item:propformal1} and~\ref{item:propformal2}, let us solve~\eqref{eqn:repDub} perturbatively at $\infty$. First, we apply the change of variables $\hat{Y} = \Psi dy$, where $\Psi$ is defined as in Proposition~\ref{prop:canonical}, and obtain
	\begin{align}
	\hat{Y}_{\zeta} = \left( U + \frac{1}{\zeta} V \right) \hat{Y},
	\end{align}
	where $V = \Psi \cV \Psi^{-1}$ and $U = \Psi \cU \Psi^{-1}$, which takes the diagonal form~\eqref{eqn:diagU}. We propose an Ansatz of the form
	\begin{align}
	\hat{Y}_i^{\textrm{formal}} = e^{\zeta u_i} \left( \hat{Y}_i^0 + \hat{Y}_i^1 \zeta^{-1} + \dots \right), \qquad \hat{Y}_i^k \in \cH(S^1) \oplus \C^n \oplus \C^{\bar{n}},
	\end{align}
	which yields the recursion
	\begin{align}
	&(u_i - U) (\hat{Y}_i^0 )= 0, \\
	&(u_i - U) (\hat{Y}^{k+1}_i) = (k+V) (\hat{Y}_i^k ), \qquad k \geq 0.
	\end{align}
	We will show that these equations have a unique solution up to normalization. Let $e_i$ denote the canonical basis vector which is $1$ at the $i$-th entry and $0$ everywhere else. Then $\ker(u_i - U) = \langle (0, e_i, 0) \rangle$, so we choose 
	\begin{align}
	\hat{Y}^0_i = (0, -z_i^2 \lambda''(z_i) e_i, 0),
	\end{align}
	which satisfies
	\begin{align}
	\Psi^{-1}(\hat{Y}^0_i) = du_i.
	\end{align}
	Let us move on to the next equation, namely
	\begin{align}
	(u_i - U) (\hat{Y}_i^1) = V(\hat{Y}_i^0). \label{eqn:formaldiscrete1}
	\end{align}
	This equation is solvable if and only if $V(0, e_i, 0) \in \textrm{im}(u_i - U)$. Note the diagonal form of $U$~\eqref{eqn:diagU} allows us to decompose the space as $\cH(S^1) \oplus \C^n \oplus \C^{\bar{n}} = \ker(u_i - U) \oplus \textrm{im}(u_i - U)$, so it is enough to show that the projection of $V(0, e_i, 0)$ to the subspace $ \ker(u_i - U)$ is $0$. We prove it as an auxiliary lemma:
	\begin{lemma}
		The operator 
		\begin{align}
		V= \Psi \cV \Psi^{-1}: \cH(S^1) \oplus \C^n \oplus \C^{\bar{n}} \longrightarrow \cH(S^1) \oplus \C^n \oplus \C^{\bar{n}}
		\end{align}
		satisfies
		\begin{align}
		&\mathbb{P}_i \circ V(0,e_i, 0) = 0, \qquad &i = 1, \dots, n \label{eqn:Vantidiag1} \\ &\overline{\mathbb{P}}_j \circ V(0,0,e_j) = 0, \qquad &j = 1, \dots, \bar{n}, \label{eqn:Vantidiag2}
		\end{align}
		where $\mathbb{P}_i$ is the projection to the $i$-th entry of the second component, and $\overline{\mathbb{P}}_j$ is the projection to the $j$-th entry of the third component.
	\end{lemma}
	\begin{proof}
		The proofs of~\eqref{eqn:Vantidiag1} and~\eqref{eqn:Vantidiag2} are analogous, so we only perform the former. First, we compute
		\begin{align}
		\Psi^{-1}(0,e_i, 0) &= \frac{1}{z_i \lambda''(z_i)} \left( \frac{zw^{\prime}(z)}{z_i-z}, -\frac{e^u}{z_i^2}, -\frac{1}{z_i} \right) = -\frac{1}{z_i^2 \lambda''(z_i)} du_i, \\
		\hat{Y} = \cV \Psi^{-1}(0,e_i, 0) &= \frac{1}{z_i \lambda''(z_i)} \left( \frac12 \frac{z w^{\prime}(z)}{z_i-z} + \frac{z w(z)}{(z_i-z)^2}, \frac12 \frac{e^u}{z_i^2}, -\frac12 \frac{1}{z_i} \right).
		\end{align}
		To conclude the proof, we have to show that
		\begin{align}
		\langle du_i, \hat{Y} \rangle = - Y_{\leq 0} (z_i) + Y_v + \frac{e^u}{z_i} Y_u = 0.
		\end{align}
		It is clear that $Y_v + \frac{e^u}{z_i} Y_u = 0$. Let us compute
		\begin{align}
		Y_{\leq 0}(z_i) = \frac{1}{z_i \lambda''(z_i)} \left( \frac12 \frac{z_i}{2 \pi \ii} \oint_{|x| = 1} \frac{w^{\prime}(x)}{(z_i-x)^2} dx + \frac{z_i}{2 \pi \ii} \oint_{|x| = 1} \frac{w(x)}{(z_i - x)^3} dx \right) = 0,
		\end{align}
		where we have used integration by parts.
	\end{proof}
As an immediate corollary, $V(\hat{Y}_i^0) \in \textrm{im}(u_i - U)$ and equation~\eqref{eqn:formaldiscrete1} admits solutions for $\hat{Y}_i^1$. Regarding uniqueness, it is clear that two different solutions of~\eqref{eqn:formaldiscrete1} must differ by an element of $\ker(u_i - U)$. Therefore, we write
\begin{align}
\hat{Y}^1_i = \hat{R}^1_i + a^1_i \hat{Y}^0_i, \qquad \hat{R}^1_i \in \textrm{im}(u_i - U), \ a_i^1 \in \C.
\end{align}
Consider the second equation
\begin{align}
(u_i - U)(\hat{Y}^2_i) = (1+V)(\hat{Y}^1_i), \label{eqn:formaldiscrete2}
\end{align}
which has a solution for $\hat{Y}^2_i$ if and only if $(1+V)(\hat{Y}^1_i) \in \textrm{im}(u_i - U)$, which happens when
\begin{align}
a_i^1 \hat{Y}_i^0 + \mathbb{P}_i \circ V (\hat{R}^1_i) = 0.
\end{align}
This fixes uniquely the constant $a_i^1$ and ensures that~\eqref{eqn:formaldiscrete2} has a solution, which, as before, can be written as $\hat{Y}^2_i = \hat{R}^2_i + a_i^2 \hat{Y}_i^0$. Iterating gives a unique $\hat{Y}_i^{\textrm{formal}}$, which concludes the proof of item~\ref{item:propformal1} of the proposition. The proof of item~\ref{item:propformal2} is completely analogous, and we will not do it explicitly here.

To prove item~\ref{item:propformal3}, we insert the Ansatz~\eqref{eqn:Ansatzformal} in the Dubrovin equation~\eqref{eqn:generalizedDubrovin} and obtain the following recursion for the functionals $r^k_p$:
\begin{align}
&\langle r_p^0, (u_p - \cU)\hat{X} \rangle = 0, \qquad &\forall \hat{X} \in \cH(S^1) \oplus \C^2 \label{eqn:rprecursion0} \\
&\langle r_p^{k+1}, (u_p - \cU) \hat{X} \rangle = \langle r_p^k, (k- \cV) \hat{X} \rangle, \qquad &\forall \hat{X} \in \cH(S^1) \oplus \C^2, \ k \geq 0 \label{eqn:rprecursionk}
\end{align}
Equation~\eqref{eqn:rprecursion0} is the eigenspace equation for the eigenvalue $u_p$. By the results of Section~\ref{sec:canonical}, we have $r_p^0 = du_p$. Note that $u_p - \cU$ is injective (one can directly see this from the diagonal form~\eqref{eqn:diagU}, noting~\ref{item:T5} excludes the degenerate case of all the canonical coordinates $u_p$ being equal), but it fails to be surjective, as
\begin{align}
\textrm{im}(u_p - \cU) = \{ \hat{Y} \in \cH(S^1) \oplus \C^2| \ \langle du_p, \hat{Y} \rangle = 0 \}
\end{align}
is a subspace of $\cH(S^1) \oplus \C^2$ of codimension 1. Therefore, the operator $u_p - \cU$ admits left-inverses. Let $B_p, B^{\prime}_p$ be two left-inverses of $u_p - \cU$. Then
\begin{align}
(B_p - B^{\prime}_p) (\hat{X}) &= (B_p - B^{\prime}_p)  \left(\left( \hat{X} - \langle du_p, \hat{X} \rangle (0,1,0) \right) + \langle du_p, \hat{X} \rangle (0,1,0) \right) \\
&= \langle du_p, \hat{X} \rangle (B_p - B^{\prime}_p) (0,1,0), \notag
\end{align}
where we have used that $\hat{X} - \langle du_p, \hat{X} \rangle (0,1,0) \in \textrm{im}(u_p - \cU)$. Therefore, a left-inverse of $u_p - \cU$ is completely determined by its action on $(0,1,0)$, and we choose $A_p$ to be the one with $A_p (0,1,0) = 0$. Back to the system~\eqref{eqn:rprecursionk}, it is now clear that the recursively defined functionals
\begin{align}
&\hat{r}_p^0 = du_p, \\
&\langle \hat{r}_p^{k+1}, \hat{X} \rangle = \langle \hat{r}_p^k, (k-\cV) A_p \hat{X} \rangle
\end{align}
solve it. Written in terms of transpose operators, the functionals
\begin{align}
\hat{r}_p^k = A_p^* (k-1-\cV)^* A_p^* (k-2-\cV)^* \dots A_p^* (-\cV)^* du_p
\end{align}
give a formal solution of the form~\eqref{eqn:Ansatzformal} to equation~\eqref{eqn:generalizedDubrovin}. Let us now study the uniqueness of solutions. Let $r_p^0 = du_p, r_p^1, \dots, r_p^k$ be given, and suppose both $s^{k+1}_p$ and $t^{k+1}_p$ solve~\eqref{eqn:rprecursionk} for $r^{k+1}_p$. Then
\begin{align}
\langle s_p^{k+1} - t^{k+1}_p, (u_p - \cU) \hat{X} \rangle = 0,
\end{align}
so $s_p^{k+1} - t^{k+1}_p$ must be a scalar multiple of $du_p$. Thus, the most general solution of the next recursive step is
\begin{align}
r^{k+1}_p = A_p^* (k-\cV)^* r_p^k + a_p^{k+1} du_p, \label{eqn:solutionrecursivestep}
\end{align}
with $a_p^{k+1} \in \C$. From~\eqref{eqn:solutionrecursivestep} we can deduce the most general form of the functionals,~\eqref{eqn:functionalsformal}, thus completing the proof.
\end{proof}

\begin{remark}
	Since $A_p (0,1,0) = 0$ and $\langle du_p, (0,1,0) \rangle = 1$, it is easy to write the functionals of any given formal solution~\eqref{eqn:Ansatzformal} in the form~\eqref{eqn:functionalsformal} by setting 
	\begin{align}
	a_p^k = \langle r_p^k, (0,1,0) \rangle.
	\end{align}
\end{remark}
\begin{remark}
	At the special point $\hat{\lambda}_0$, we can compute the operator $A_p$ explicitly
	\begin{align}
	A_p \hat{X} = \left( \frac12 e^{-u} \left( \frac{1}{p} - \frac{1}{z} \right)^{-1} (X(z) - X(p)), \ - \frac12 \left( \frac1p X(p) + Y_u \right), \ -\frac12 e^{-u} X(p) \right)
	\end{align}
\end{remark}
\begin{remark}[Uniqueness of formal solutions]
	Let us explain why the functionals $r^k_i$ in the expansion of $\xi_{i}^{\textrm{formal}}$ are uniquely determined, whereas $r_p^k$ in the expansion of $\xi_p^{\textrm{formal}}$ are not. Assume we have $r^{0}_{i}, r^1_{i}, \dots, r^{k-1}_{i}$ and let $t^k_{i}$ be such that
	\begin{align}
	\langle t^{k}_{i}, (u_{i} - \mathcal{U}) \hat{X} \rangle = \langle r^{k-1}_{i}, (k-1-\mathcal{V}) \hat{X} \rangle .
	\end{align}
	Then the general solution of 
	\begin{align}
	\langle r^{k}_{i}, (u_{i} - \mathcal{U}) \hat{X} \rangle = \langle r^{k-1}_{i}, (k-1-\mathcal{V}) \hat{X} \rangle 
	\end{align}
	is given by $r^k_{i} = t^k_{i} + a^k_{i} du_{i}$. To fix the constant $a^k_{i}$ we consider the next equation
	\begin{align}
	\langle r^{k+1}_{i}, (u_{i} - \mathcal{U}) \hat{X} \rangle = \langle t^k_{i} + a^k_{i} du_{i}, (k-\mathcal{V}) \hat{X} \rangle ,
	\end{align}
	and choose $\hat{X}$ to be the vector representative of $du_i$, i.e., $\eta_*(\hat{X}) = du_i$ (here we do not denote $\hat{X} = du_i$ as usual because it might lead to confusion). In particular, $\hat{X}  \in \ker(u_{i} - \mathcal{U})$, which gives  
	\begin{align}
	a^k_{i} = - \frac{\langle t^k_{i}, (k-\mathcal{V}) \hat{X} \rangle}{\langle du_{i}, (k-\mathcal{V}) \hat{X} \rangle} .
	\end{align}
	Note that the denominator does not vanish since $\mathcal{V} (\hat{X}) \in \textrm{im} (u_{i} - \mathcal{U})$ and $\hat{X} \notin \textrm{im} (u_{i} - \mathcal{U})$. On the other hand, it is impossible to repeat this procedure to fix the constants appearing in $\xi_{p}^{\textrm{formal}}$, as the operator $(u_p - \mathcal{U})$ is injective.
\end{remark}

\section{Integral solutions and their asymptotics} 
\label{sec:analyticsolutions}

In this section, we find a family of solutions to the Dubrovin equation defined in terms of an exponential integral along the unit circle in the complex plane. We derive the asymptotic behaviour of such solutions at $\zeta\sim \infty$, obtaining this way formal solutions in the sense of the previous section.

\subsection{Integral solutions}

We define a family of functionals $y_\sigma(\zeta)$ on $M_0 \times \C$ and we prove explicitly that their differentials $dy_\sigma(\zeta)$ solve the Dubrovin equation~\eqref{eqn:generalizedDubrovin}.

\begin{proposition}
Let $\sigma \in \mathbb{C}$ and consider the functionals
\begin{align} \label{eq:ysol}
y_\sigma (\zeta) = \frac{\zeta^{-1/2}}{2 \pi \ii} \oint_{|z| = 1} e^{\zeta \lambda_\sigma(z) } \frac{dz}{z}.
\end{align}
Their differentials $dy_\sigma(\zeta)$ are solutions of the Dubrovin equation~\eqref{eqn:generalizedDubrovin}.
\end{proposition}
\begin{proof}
The differentials $dy_\sigma(\zeta) \in T^*_{\hat{\lambda}} M$ are given by
	\begin{align} 
	\langle dy_\sigma (\zeta), \hat{X} \rangle = \frac{\zeta^{1/2}}{2 \pi \ii} \oint_{|z| = 1} e^{\zeta \lambda_\sigma(z)} \langle d\lambda_\sigma(z), \hat{X} \rangle \frac{dz}{z}. \label{eqn:dysigmadfn}
	\end{align}
	Plugging~\eqref{eqn:dysigmadfn} in~\eqref{eqn:generalizedDubrovin} yields
	\begin{align}
	&\langle dy_\sigma (\zeta), \left( \mathcal{U} - \frac{1}{\zeta} \mathcal{V}  \right) \hat{X} \rangle - \langle dy_\sigma (\zeta), \hat{X} \rangle_\zeta = \notag \\ &= 
	\frac{\zeta^{1/2}}{2 \pi \ii} \oint_{|z| = 1} e^{\zeta \lambda_\sigma(z)} \Big( \langle d\lambda_\sigma(z), \mathcal{U} \hat{X} \rangle - \lambda_\sigma(z) \langle d\lambda_\sigma(z), \hat{X} \rangle +\\ 
	&\qquad + z \lambda_\sigma^{\prime}(z) \left\langle d\lambda_\sigma(z), \left( \frac{w(z)}{z w^{\prime}(z)} X(z), 0, - X_u \right) \right\rangle \Big) \frac{dz}{z}, \notag
	\end{align}
	which vanishes by Lemma~\ref{lem:keylemma}.
\end{proof}

\begin{remark}
The differentials $dy_\sigma(\zeta)$ are actually representable, see Remark~\ref{dyarerepresentable} below. We can therefore use Proposition 20 in~\cite{CM15} to prove that they are covariantly constant w.r.t. the full deformed flat connection $\widetilde{\nabla}$.
\end{remark}

\subsection{Asymptotics}

Let us study the asymptotics of the solutions $dy_\sigma(\zeta)$ for $|\zeta| \to \infty$. The usual approach to find the asymptotics of integrals of the form~\eqref{eq:ysol} or~\eqref{eqn:dysigmadfn} is by applying the steepest descent method, first by expressing the path of integration as a combination of the steepest descent paths passing through the critical points of the superpotential $\lambda_\sigma(z)$, and then by computing saddle point asymptotics, see e.g.~\cite{wittenAnalyticContinuationChernsimons2011}.

In our case, however, for generic values of $\sigma$, the path of integration cannot be deformed away from the domain of definition of $\lambda_\sigma(z)$, namely a neighbourhood of $S^1$. We will therefore restrict our analysis to those values of $\sigma$ such that the critical points of $\lambda_\sigma(\zeta)$ belong to $S^1$, and to $\sigma=0$ and $1$, for which $\lambda_\sigma$ coincides with $-\lambda$ and $\bla$, respectively.

Let us consider a point $\hat{\lambda} = (\lambda, \bar{\lambda})$ in the Dubrovin--Frobenius manifold $M_0$ such that $\lambda$  and $\bar{\lambda}$ have $n$ and $\bar{n}$ critical points, respectively. Denote by $z_i$ and $\bar{z}_j$ the critical points and by $u_i$ and $\bar{u}_j$ the critical values of $-\lambda$ and $\bar\lambda$, respectively, as in~\eqref{eq:uub}-\eqref{eq:uub1}. Recall that we define a curve $\Sigma$ via the function $\sigma(z)$ on $S^1$, see~\eqref{eqn:curvesigma}. For every $\sigma\in\Sigma$ the superpotential $\lambda_\sigma(z)$ has a finite number of critical points $p_1, \dots  , p_s$, which are non-degenerate because of~\eqref{eq:nondeggg}. 

For $\sigma$ belonging to the curve $\Sigma$, the path of integration $S^1$ passes through the points $p_1, \dots  , p_s$. For $|\zeta| \to \infty$ in a generic direction in the $\zeta$-plane the asymptotics of the integral will be dominated by the saddle point asymptotics of one of such points. 
More precisely, let us consider the lines in $\C$ passing through the origin and given by $\Re( \zeta (u_{p_i} - u_{p_j} )) =0$ for  $i, j = 1, \dots , s$. 
These lines divide the $\zeta$-plane in sectors $S(u_{p_j})$ for  $j=1, \dots, s$ such that in the sector $S(u_{p_j})$ the exponential $e^{\zeta u_{p_j}}$ has the dominant asymptotic behaviour as $|\zeta| \to \infty $.

For $\sigma=0$, the critical points $z_i$ of the exponent $-\lambda(z)$ belong to the exterior of the unit disc. In this case, however, the integrand is holomorphic in $D_\infty$, so we can deform the path of integration in such a way that it passes through all the critical points. As above, in each of the sectors $S(u_j)$ for $j=1, \dots , n$  determined by removing the lines $\Re( \zeta (u_i - u_j )) =0$ for  $i, j = 1, \dots , n$ from $\C$, the critical value $u_j$ will determine the asymptotics. 

Similarly, for $\sigma =1$, the path of integration can be deformed in such a way that it passes through all the critical points $\bar{z}_i$ in the interior of $D_0$ and the critical value $\bar{u}_j$ will dominate the asymptotics in a sector $S(\bar{u}_j)$, among the sectors obtained by removing the lines $\Re( \zeta (\bar{u}_i - \bar{u}_j )) =0$ for  $i, j = 1, \dots , \bar{n}$ from the $\zeta$-plane. 

For any $\hat{X} \in T_{\hat{\lambda}} M$ we have the following asymptotic behaviour.

\begin{proposition} 
\label{prop:asymptoticsdysigma} 
For $\sigma \in \Sigma$ and $p = p_j$ one of the critical points of $\lambda_\sigma$, we have 
\begin{align}
	\langle dy_\sigma, \hat{X} \rangle \sim e^{\zeta u_p } \sum_{k=0}^{\infty} \frac{1}{2 \Gamma \left( \frac12 - k \right)} \frac{1}{2 \pi \ii} \oint_p \frac{\langle d\lambda_\sigma (z), \hat{X} \rangle }{( \lambda_\sigma(z) - u_p)^{k + \frac12} } \frac{dz}{z} \ \zeta^{-k}, \ |\zeta| \to \infty , \ \zeta \in S(u_p).
\label{eqn:dypformal}
\end{align}
For $\sigma = 0$ and $z_j$ one of the critical points of $\lambda$, we have 
\begin{align}
	\langle dy_0, \hat{X} \rangle \sim -e^{ \zeta u_j } \sum_{k=0}^{\infty} \frac{1}{2 \Gamma \left( \frac12 - k \right)} \frac{1}{2 \pi \ii} \oint_{z_j} \frac{\langle d\lambda(z) , \hat{X} \rangle }{( - \lambda(z) - u_j )^{k + \frac12} } \frac{dz}{z} \ \zeta^{-k}, 
	 \ |\zeta| \to \infty , \ \zeta \in S(u_j).
	 \label{eqn:dy0formal}
\end{align}
For $\sigma = 1$ and $\bar{z}_j$ one of the critical points of $\bar\lambda$, we have 
\begin{align}
	\langle dy_1, \hat{X} \rangle \sim e^{ \zeta \bar{u}_j } \sum_{k=0}^{\infty} \frac{1}{2 \Gamma \left( \frac12 - k \right)} \frac{1}{2 \pi \ii} \oint_{\bar{z}_j} \frac{\langle d\bar{\lambda}(z) , \hat{X} \rangle }{( \bar{\lambda}(z) - \bar{u}_j )^{k + \frac12} } \frac{dz}{z} \ \zeta^{-k}, 
	 \ |\zeta| \to \infty , \ \zeta \in S(\bar{u}_j).
	\label{eqn:dy1formal}
\end{align}
\end{proposition}

In the above formulas the symbol $\oint_z$ denotes integration along a small counterclockwise simple path around $z$.

\begin{proof}
	Expression~\eqref{eqn:dypformal} follows from Lemma~\ref{lem:saddlepoint} applied to $\langle dy_\sigma, \hat{X} \rangle$. For~\eqref{eqn:dy0formal}, note the integrand of $\langle dy_0, \hat{X} \rangle$ is holomorphic on $D_\infty \setminus \{ \infty \}$, so we can deform the path of integration to one that passes through all critical points of $\lambda$, and then apply again Lemma~\ref{lem:saddlepoint}. Finally, for~\eqref{eqn:dy1formal}, we deform the path so that it passes through all critical points of $\bar{\lambda}$, and then we apply Lemma~\ref{lem:saddlepoint}.
\end{proof}

\begin{proposition}
	The asymptotic expansions of $\{dy_\sigma\}_{\sigma \in \Sigma \cup \{ 0, 1 \} }$ at $|\zeta| \rightarrow \infty$ given in Proposition~\ref{prop:asymptoticsdysigma} are formal solutions of the Dubrovin equation.
\end{proposition}

\begin{proof}
	Let us prove that~\eqref{eqn:dypformal} defines a formal solution of the Dubrovin equation of the form~\eqref{eqn:Ansatzformal} with $r^k_p$ given by
	\begin{align}
	\langle r_p^k, \hat{X} \rangle =  \frac{1}{2 \Gamma \left( \frac12 - k \right)} \frac{1}{2 \pi \ii} \oint_p \frac{\langle d\lambda_\sigma (z), \hat{X} \rangle }{( \lambda_\sigma(z) - u_p)^{k + \frac12} } \frac{dz}{z}.
	\end{align}
	For that, we need to prove that the cotangent vectors $r_p^k$ satisfy the recursion relations~\eqref{eqn:rprecursion0}-\eqref{eqn:rprecursionk}. Let us first show~\eqref{eqn:rprecursion0}. We have
\begin{align}
	\langle r_p^0, (\mathcal{U} - u_p) \hat{X} \rangle  &= \frac{1}{2 \Gamma\left( \frac12 \right)} \frac{1}{2 \pi \ii} \oint_p \frac{\langle d\lambda_\sigma(z), \mathcal{U} \hat{X} \rangle - u_p \langle d\lambda_\sigma(z), \hat{X} \rangle}{(\lambda_\sigma(z) - u_p)^{\frac12}} \frac{dz}{z}. 
\end{align}
By adding and subtracting a term proportional to $\lambda_\sigma(z)$, the previous expression equals
\begin{align}  \label{eqn:proofpropasymp}
	&\frac{1}{2 \Gamma\left( \frac12 \right)} \frac{1}{2 \pi \ii} \oint_p \frac{\langle d\lambda_\sigma(z), \mathcal{U} \hat{X} \rangle - \lambda_\sigma(z) \langle d\lambda_\sigma(z), \hat{X} \rangle}{(\lambda_\sigma(z) - u_p)^{\frac12}} \frac{dz}{z} +\\
	&+ \frac{1}{2 \Gamma\left( \frac12 \right)} \frac{1}{2 \pi \ii} \oint_p \frac{ (\lambda_\sigma(z) -  u_p) \langle d\lambda_\sigma(z), \hat{X} \rangle}{(\lambda_\sigma(z) - u_p)^{\frac12}} \frac{dz}{z}.  \notag
	\end{align}
	The second summand equals
	\begin{align}
	\frac{1}{2 \Gamma\left( \frac12 \right)} \frac{1}{2 \pi \ii} \oint_p (\lambda_\sigma(z) -  u_p)^{1/2} \langle d\lambda_\sigma(z), \hat{X} \rangle \frac{dz}{z},
	\end{align}
	which vanishes because the integrand is holomorphic at $p$. By Lemma~\ref{lem:keylemma}, 
	\begin{align}
	\langle d\lambda_\sigma(z), \mathcal{U} \hat{X} \rangle - \lambda_\sigma(z) \langle d\lambda_\sigma(z), \hat{X} \rangle = (g(z) + C) z \lambda_\sigma^{\prime}(z),
	\end{align}
	where $C$ is a constant and 
	\begin{align}
	g(z) = -\left\langle d\lambda_\sigma(z), \left( \frac{w(z)}{z w^{\prime}(z) } X(z), 0, -X_u  \right) \right\rangle 
	\end{align}
	is holomorphic at $p$. Therefore, the first summand of~\eqref{eqn:proofpropasymp} equals
	\begin{align}
	\frac{1}{2 \Gamma\left( \frac12 \right)} \frac{1}{2 \pi \ii} \oint_p \frac{(g(z) + C) \lambda_\sigma^{\prime}(z)}{(\lambda_\sigma(z) - u_p)^{\frac12}} dz = - \frac{1}{ \Gamma\left( \frac12 \right)} \frac{1}{2 \pi \ii} \oint_p g^{\prime}(z)(\lambda_\sigma(z) - u_p)^{\frac12} dz,
	\end{align}
	which again vanishes by holomorphicity at $p$ of the integrand. 
To prove~\eqref{eqn:rprecursionk}, observe that by a computation similar to the previous one, we can write 
	\begin{multline}
	\langle r_p^{k+1}, (\mathcal{U} - u_p) \hat{X} \rangle - \langle r_p^k, (k - \mathcal{V}) \hat{X} \rangle = 
 \frac{1}{2 \Gamma\left( -k - \frac12 \right)} \times \\
 \times \frac{1}{2 \pi \ii} \oint_p \frac{\langle d\lambda_\sigma(z), \mathcal{U} \hat{X} \rangle - \lambda_\sigma(z) \langle d\lambda_\sigma(z), \hat{X} \rangle + z \lambda_\sigma^{\prime}(z) \left\langle d\lambda_\sigma(z), \left( \frac{w(z)}{z w^{\prime}(z)} X(z), 0, - X_u \right) \right\rangle}{(\lambda_\sigma(z) - u_p)^{k+\frac32}} \frac{dz}{z},
 	\end{multline}
	which vanishes by Lemma~\ref{lem:keylemma}. The proofs for~\eqref{eqn:dy0formal} and~\eqref{eqn:dy1formal} are completely analogous.
\end{proof}

\begin{remark} 
\label{dyarerepresentable}
The 1-forms~\eqref{eqn:dysigmadfn} are representable for any $\sigma \in \mathbb{C}$, i.e. $dy_\sigma(\zeta) \in T^*_{\hat{\lambda}} M^{\mathrm{rep}}$, with representative in the tangent given by	
\begin{align}
dy_\sigma = \zeta^{1/2} \Bigg( &\sigma z w'(z) e^{\zeta \lambda_\sigma (z)} - z w'(z) \left( e^{\zeta \lambda_\sigma(z)} \right)_{\geq 0},\\
&\frac{1}{2 \pi \ii} \oint_{|z| = 1} e^{\zeta \lambda_\sigma(z)} \frac{e^u}{z} \frac{dz}{z} ,  \frac{1}{2 \pi \ii} \oint_{|z| = 1} e^{\zeta \lambda_\sigma(z)} \frac{dz}{z}\Bigg). \notag
\end{align}
Notice, however, that the functionals $r^k_p$ in the asymptotic expansion are in general not representable, in particular the leading term $r^0_p$ is proportional to the non-representable functional $du_p$.
\end{remark}

\begin{remark} \label{remark:noncomplete}
	The family of solutions $\{ dy_\sigma (\zeta) \}_{\sigma \in \Sigma \cup \{ 0,1 \} }$ is not complete. For example, at the special point $\hat{\lambda}_0$, the tangent vector
	\begin{align}
	\hat{X} = \left( \left( 1 - e^{-\zeta \frac{e^u}{z} } \right)z, 0, -1 \right)
	\end{align}
	satisfies $	\langle dy_\sigma(\zeta), \hat{X} \rangle = 0$ for all $\sigma$.
\end{remark}

\begin{remark} \label{remark:monodromy}
	The monodromy of the solutions $dy_\sigma (\zeta)$ is trivial since it just originates from the $\zeta^{1/2}$ factor
\begin{equation}
dy_\sigma(\zeta e^{2 \pi \ii}) = - dy_\sigma (\zeta).
\end{equation}
\end{remark}

\section{Resurgence and Stokes phenomena} 
\label{sec:resurgenceandstokes}

In this section, we study the Stokes phenomenon at the irregular singularity $\zeta \sim \infty$ of the Dubrovin equation. 

In the finite dimensional case~\cite{Dub99}, the Dubrovin equation (written in the normalized canonical frame) has a unique formal fundamental solution of the following form
\begin{equation}
Y_{\mathrm{formal}}(\zeta) =  (\mathrm{Id} + R_1 \zeta^{-1} + R_2 \zeta^{-2} + \dots ) e^{\zeta U},
\end{equation}
where $U =  \diag(u_1, \cdots, u_n)$ with $u_i \not =u_j$ for $i\not = j$.
An admissible line $\ell$ through the origin in $\C$ is given by the choice of its positive direction $\phi$ such that it satisfies $\Re (e^{\ii \phi } (u_i -u_j)) \not= 0$ for any $i\not= j$. 
It can be shown that, given a choice of admissible line $\ell$ in $\C$, there exists a unique fundamental solution $Y_{\mathrm{right}}$ (resp. $Y_{\mathrm{left}}$) which is asymptotic to $Y_{\mathrm{formal}}$ for $\zeta \sim \infty$ on the open sector $\Pi^\epsilon_{\textrm{right}}$ (resp. $\Pi^\epsilon_{\mathrm{left}}$) of opening slightly larger than $\pi$ containing the right (resp. left) half-planes separated by $\ell$. 
The Stokes matrices $S_\pm$ relate such fundamental solutions on the intersection $\Pi^\epsilon_{\textrm{right}} \cap \Pi^\epsilon_{\textrm{left}} = \Pi^\epsilon_+ \cup \Pi^\epsilon_-$, namely
\begin{equation} \label{eq:stm}
Y_L(\zeta) =Y_R(\zeta) S_\pm, \qquad \zeta\in \Pi^\epsilon_\pm,
\end{equation}
where $\Pi^\epsilon_+$, resp. $\Pi^\epsilon_-$, is the sector containing the direction $\phi$, resp. $\phi + \pi$.

Notice that the columns of a fundamental solution give a basis of  solutions of the Dubrovin equation. In the infinite-dimensional case, we might consider the family of integral solutions $\{ dy_\sigma \}_{\sigma\in \Sigma \cup \{0,1\} }$ obtained in the previous section. 
Such family, however, is not complete and moreover has trivial monodromy, see Remarks~\ref{remark:noncomplete} and~\ref{remark:monodromy}, therefore it cannot be used to obtain the analogues of the Stokes matrices.
To find a larger family of solutions we adopt a different strategy, using resurgence theory to associate a family of ``weak'' solutions to a family of formal solutions like those studied in Section~\ref{sec:formalsolutions}. 
More precisely, we consider the family of formal solutions given by the asymptotic expansions of the integral solutions and we apply to it the Borel resummation procedure.

Resurgence theory~\cite{hardyDivergentSeries1973, ecalleFonctionsResurgentesVol1981, ramis1993series, balserDivergentPowerSeries1994} provides a method to associate analytic functions to formal series which are not convergent. The resummation procedure of a formal power series $\varphi(\zeta) = \sum_{k \geq0 } a_k \zeta^{-k}$ can be summarized, for our aims, in three steps: computation of the sum of its Borel transform $\widehat{\varphi}(\chi)$  obtained by the substitution $\zeta^{-k} \mapsto \chi^k/k!$, analytic continuation and identification of the resurgent structure of $\widehat{\varphi}(\chi)$, namely of its behaviour at singular points, and resummation to a function $s_\theta(\varphi)(\zeta)$ via  Laplace transform. 
For recent expositions of these methods we refer the reader to~\cite{Mar21, Sau14, Dor19, MSS04}.

For simplicity, we restrict to the special point $\hat{\lambda}_0$ in $M_0$. We also require $|e^u| <1$ so that there are no discrete canonical coordinates to consider.

\subsection{Weak solutions}

Recall that the cotangent space $T^*_{\hat{\lambda}}M$ at a point $\hat{\lambda} \in M_0$ is given by the algebraic dual of $T_{\hat{\lambda}}M$. Given a cotangent vector $\xi \in T^*_{\hat{\lambda}}M$, we define its coefficients as the numbers $\langle \xi, e_{\hat{m}} \rangle$ obtained by acting on the elements $e_{\hat{m}} \in \cH(S^1) \oplus \C^2$, given by 
\begin{equation} \label{eq:coeffs}
e_m =  (z^m,0,0), \quad  e_v=(0,1,0), \quad  \text{and} \quad  e_u =(0,0,1),
\end{equation}
where $\hat{m} \in \Z \cup \{v, u\}$. 
In general, an arbitrary choice of coefficients $C_{\hat{m}}$ does not define a cotangent vector $\xi$ with $C_{\hat{m}} = \langle \xi, e_{\hat{m}} \rangle$. 
However, it always defines an element in $T^*_{\hat{\lambda}_0} M^{\mathrm{weak}}$, which is the algebraic dual of 
\begin{align*}
		T_{\hat{\lambda}_0} M^{\mathrm{test}} = \{ \hat{X} = (X(z), X_v, X_u) \in T_{\hat{\lambda}_0} M | \ X(z) \in \C[z, z^{-1}] \} \cong \C[z, z^{-1}] \oplus \C^2,
\end{align*}
namely the subset of $T_{\hat{\lambda}_0} M$ consisting of Laurent polynomials in $z$, by the formula
\begin{align}
\left\langle \xi, (X(z), X_v, X_u) \right\rangle = \sum_n X_n C_{n} + X_v C_v + X_u C_u
\end{align}
for $X(z) = \sum X_n z^n \in \C[z, z^{-1}]$.

The motivation for introducing $T^*_{\hat{\lambda}_0} M^{\mathrm{weak}}$ is that we are going to obtain solutions to the Dubrovin equation by Borel resummation of the coefficients of the formal integral solutions, which will turn out not to be in $T^*_{\hat{\lambda}_0} M$.
Notice that, since the operators $\cU$ and $\cV$ at the special point preserve the subspace $T_{\hat{\lambda}_0} M^{\mathrm{test}}$, it is possible to define weak solutions to the Dubrovin equation~\eqref{eqn:generalizedDubrovin}, i.e., $\xi = \xi (\zeta) \in T^*_{\hat{\lambda}_0} M^{\mathrm{weak}}$ such that
\begin{align}
\langle \xi, \hat{X} \rangle_{\zeta} = \langle \xi, \left( \cU - \frac{1}{\zeta} \cV \right) \hat{X} \rangle, \qquad \forall \hat{X} \in T_{\hat{\lambda}_0} M^{\mathrm{test}}.
\end{align}

\subsection{Formal integral solutions} 

At the special point, we can give an explicit formula for the coefficients of the formal solutions corresponding to the asymptotic expansion of the integral solutions obtained in Section~\ref{sec:analyticsolutions}, namely (recall Proposition~\ref{prop:asymptoticsdysigma})
\begin{align} \label{eq:frml}
\langle dy_p^{\mathrm{formal}}, \hat{X} \rangle = e^{\zeta u_p } \sum_{k=0}^{\infty} \frac{1}{2 \Gamma \left( \frac12 - k \right)} \frac{1}{2 \pi \ii} \oint_p \frac{\langle d\lambda_\sigma (z), \hat{X} \rangle }{( \lambda_\sigma(z) - u_p)^{k + \frac12} } \frac{dz}{z} \ \zeta^{-k}.
\end{align}

\begin{lemma}
The coefficients of the formal integral solutions are given by
\begin{align}
\langle dy_p^{\textrm{formal}}, e_m \rangle 
&= e^{\zeta u_p} \frac{1}{2} \sigma(p) p^m  \sqrt{\frac{p}{e^u}}   \varphi_p^m(\zeta) \label{eqn:dyformal1}, \qquad &m \geq 1 \\
\langle dy_p^{\textrm{formal}}, e_m \rangle 
&= e^{\zeta u_p}  \frac{1}{2} (\sigma(p) - 1) p^m  \sqrt{\frac{p}{e^u}}  \varphi_p^m(\zeta), \label{eqn:dyformal2} \qquad &m \leq 0 \\
\langle dy_p^{\textrm{formal}}, e_v \rangle 
&= e^{\zeta u_p}  \frac{1}{2} \sqrt{\frac{p}{e^u}}  \varphi_p^0(\zeta), \label{eqn:dyformal3} \\
\langle dy_p^{\textrm{formal}}, e_u \rangle 
&=  e^{\zeta u_p} \frac{1}{2}  \sqrt{\frac{e^u}{p}}  \varphi_p^{-1}(\zeta), \label{eqn:dyformal4}
\end{align}
where 
\begin{align} 
\varphi_p^m(\zeta) = \sum_{k=0}^\infty \frac{1}{\Gamma\left(\frac12-k\right)} \binom{m+k-1/2}{2k} \left( \frac{p}{\zeta e^u} \right)^k .\label{eqn:nowhereconvfps}
\end{align}
\end{lemma}

\begin{proof}
The coefficients in~\eqref{eq:frml} are proportional to the integrals
\begin{align}
I_{k,m} = \frac{1}{2 \pi \ii} \oint_p \frac{z^m}{(\lambda_{\sigma(p)} (z) - \lambda_{\sigma(p)} (p))^{k+\frac12} } \frac{dz}{z}.
\end{align}
Since at the special point $\sigma(p) = 1 + \frac{e^u}{p^2}$, we have that
\begin{align}
\lambda_{\sigma(p)} (z) - \lambda_{\sigma(p)} (p) = \frac{e^u}{z} \left( \frac{z}{p}  - 1\right)^2,
\end{align}
therefore 
\begin{align}
I_{k,m} &= \frac{p^{2k+1}}{e^{(k+1/2)u}} \ \frac{1}{2 \pi \ii} \oint_p \frac{z^{m+k-1/2}}{(z-p)^{2k+1} } dz = \frac{p^{2k+1}}{e^{(k+1/2)u}} \ \frac{1}{2 \pi \ii} \oint_0 \frac{(z+p)^{m+k-1/2}}{z^{2k+1} } dz.
\end{align}
Expanding the numerator and computing the residue at 0 gives
\begin{align}
I_{k,m} = \binom{m+k-1/2}{2k} \left( \frac{p}{e^u} \right)^{k+1/2} p^m,
\end{align}
from which the lemma follows immediately.
\end{proof}

\begin{remark}
Notice that~\eqref{eqn:nowhereconvfps} is a nowhere convergent formal power series in $\zeta^{-1}$. 
\end{remark}

\subsection{Borel transform and resurgent structure}

Recall that a formal power series $\varphi(\zeta) = \sum_{k \geq0 } a_k \zeta^{-k}$ at $\zeta \sim \infty$ is called Gevrey-$1$ if $|a_k| \leq C^k k!$ for all $k > 0$ for some positive constant $C$. 
In such case, its Borel transform, namely the series 
\begin{equation}
\widehat{\varphi}(\chi) = \sum_{k\geq0} \frac{a_k}{k!} \chi^k ,
\end{equation}
is convergent in a neighbourhood of $\chi \sim 0$. 

The formal integral solutions $\varphi_p^m(\zeta)$ are are clearly Gevrey-$1$ and we can explicitly identify their Borel transform: 

\begin{proposition}
	The Borel transform of $\varphi^m_p(\zeta)$ converges for $|\chi| \leq \left| \frac{4e^u}{p}\right|$ and is given by
	\begin{align}
	\widehat{\varphi}^m_p(\chi) = \frac{1}{\sqrt{\pi}} \ {}_2F_1 \left( \frac12 - m, \frac12 + m; 1; \frac{p \chi}{4 e^u} \right), \label{eqn:Boreltransform}
	\end{align}
	where ${}_2F_1(a,b;c;z)$ denotes the Gauss hypergeometric function.
\end{proposition}

\begin{proof} Applying $\zeta^{-k} \rightarrow \chi^k/k!$ to~\eqref{eqn:nowhereconvfps} yields
\begin{align}
\widehat{\varphi}^m_p (\chi) &= \sum_{k=0}^\infty \frac{1}{k! \Gamma\left(\frac12-k\right)} \binom{m+k-1/2}{2k} \left( \frac{p\chi}{e^u} \right)^k \label{eqn:initialBoreltransform} \\
&= \frac{1}{\sqrt{\pi}} \sum_{k=0}^\infty \binom{-1/2}{k} \binom{m+k-1/2}{2k} \left( \frac{p\chi}{e^u} \right)^k .\notag
\end{align}
The desired result follows immediately from  Lemma~\ref{lemma:aux-hg}, see Appendix~\ref{sec:hypergeometric}.
\end{proof}

Let us now consider the so-called resurgent structure of the Borel transform. 
The Borel transform $\widehat{\varphi}^m_p (\chi)$ has a singularity at $\chi_p = 4e^u/p$ corresponding to the logarithmic branch point at $z=1$ of the hypergeometric function ${}_2F_1(a,b;a+b;z)$, see~\eqref{eqn:hypergeometriclog}. 
Near the singularity it takes the form
\begin{align} \label{eqn:alienderivative}
\widehat{\varphi}^m_p (\chi_p + \xi) &=
 \frac{   (-1)^{m+1}}{\pi } 
 \log\left(-\frac{p \xi}{4e^u}\right)  \frac{1}{\sqrt{\pi}} 
 \ {}_2F_1\left( \frac12 - m, \frac12 + m, 1; - \frac{p \xi}{4 e^u} \right) + 
f_{\mathrm{reg}}(\xi)
\end{align}
where $f_{\mathrm{reg}}(\xi)$ is holomorphic near $\xi \sim 0$. Here we have used the identity
\begin{align}
\Gamma \left( \frac12 - m \right) \Gamma \left( \frac12 + m \right) = \frac{\pi}{\sin \left( \frac12 + m \right) \pi} = (-1)^m \pi.
\end{align}

It is important to notice that the function multiplying the logarithm 
\begin{align}
\widehat{\varphi}^m_{-p}(\xi) = \frac{1}{\sqrt{\pi}} \ {}_2F_1\left( \frac12 - m, \frac12 + m, 1; - \frac{p \xi}{4 e^u} \right),
\end{align}
is actually the Borel transform of  the formal solution $\varphi^m_{-p} (\zeta)$ above for a different sign of $p$.

\subsection{Borel resummation} 

The Borel resummation $s_\theta(\varphi)(\zeta)$ of the formal series $\varphi(\zeta)$ is defined as the Laplace transform of its Borel transform
\begin{equation}
s_\theta(\varphi)(\zeta) = \zeta \int_{C^\theta} \widehat{\varphi}(\chi) e^{-\zeta \chi} d\chi,
\end{equation}
where the integral is along a ray $C^\theta = e^{\ii \theta} \mathbb{R}_{+}$ that does not contain any singularity of $\widehat{\varphi}(\chi)$. 
The function $s_\theta(\varphi)(\zeta)$ is holomorphic on the sector in $\C^*$ given by those $\zeta$ such that $| e^{-\zeta \chi}| \to 0$ for $\chi\to \infty$ along $C^\theta$ and it is asymptotic to the formal series $\varphi(\zeta)$ for $\zeta \sim \infty$. The above integral representation of the (possibly multivalued) analytic continuation of $s_\theta(\varphi)(\zeta)$ also holds outside the sector, provided the path of integration is deformed accordingly.

Denote $C^{\theta_{\mathrm{St}}}$ the ray passing through the logarithmic singularity of $\widehat{\varphi}_p^m (\chi)$ at  $\chi_p = 4e^u/p$, corresponding to $\theta_{\textrm{St}} = \arg e^u - \arg p$.
For any ray $C^\theta = e^{\ii \theta} \mathbb{R}_{+}$ with $\theta \not= \theta_{\textrm{St}}$, the Borel resummation 
\begin{align}
s_\theta (\varphi_p^m) (\zeta) = \zeta \int_{C^{\theta}} \widehat{\varphi}_p^m (\chi) e^{-\zeta \chi} d\chi
\end{align}
defines an analytic function in the sector where the real part of the exponential is negative, i.e., the half-plane
\begin{align} \label{eq:pt}
\Pi_\theta = \left\{ \zeta \in \mathbb{C} | \ -\theta - \frac{\pi}{2} < \arg \zeta < -\theta + \frac{\pi}{2} \right\}.
\end{align}
Moreover,
\begin{align}
s_\theta (\varphi^m_p)(\zeta) \sim \varphi^m_p(\zeta),\label{eqn:borelintegralasymp}
\end{align}
for $|\zeta| \rightarrow \infty$ in the sector $\Pi_\theta$.

Denote $s(\varphi_p^m)(\zeta)$ the multivalued analytic continuation of $s_\theta(\varphi_p^m)(\zeta)$ on $\C^*$. Notice that one obtains the same function by analytically continuing $s_{\theta'}(\varphi_p^m)(\zeta)$ for $\theta' \not=\theta$ in the appropriate direction.

Observe that $s(\varphi_p^m)(\zeta)$ is asymptotic to $\varphi^m_p(\zeta)$ for  $|\zeta| \rightarrow \infty$ in any sector where it is given by an integral representation as above. 
Denoting $\theta_0 = \theta_{\mathrm{St}} +\pi $, this happens whenever $\theta \not= \theta_{\mathrm{St}}$, namely when $\theta \in ( \theta_0 - \pi, \theta_0 + \pi )$. 
Therefore 
\begin{equation} \label{eq:asp}
s (\varphi^m_p)(\zeta) \sim \varphi^m_p(\zeta)
\end{equation}
for 
\begin{equation}
\zeta \in \bigcup_{\theta \in ( \theta_0 - \pi, \theta_0 + \pi )} \Pi_\theta,
\end{equation}
i.e. when $\zeta$ belongs to the sector of opening $3\pi$ given by
\begin{equation}
-\theta_0 - \frac{3\pi}{2} < \arg\zeta< -\theta_0 + \frac{3\pi}{2}.  
\end{equation}

The monodromy of the multivalued function $s(\varphi_p^m)$ is determined by the resurgent structure of the Borel transform. Indeed, for $\zeta \in \Pi_{\theta_\mathrm{St}}$ we have
\begin{equation}
s(\varphi_p^m) (e^{2\pi \ii}\zeta) -s(\varphi_p^m) (\zeta) 
= \zeta \int_{\mathcal{H}} \widehat{\varphi}_p^m (\chi) e^{-\zeta \chi} d\chi ,
\end{equation}
where $\cH$ is the clockwise Hankel contour around the singular point $\chi_p = 4e^u/p$ coming from infinity along the direction $\theta_{\mathrm{St}}$. By substituting~\eqref{eqn:alienderivative} and performing the change of variable of integration $\chi = \chi_p + \xi$, we find it equals
\begin{equation}
-(-1)^m  e^{-\zeta \frac{4 e^u}{p}} \zeta \frac{1}{\pi} \int_{\mathcal{H}_0} \log \left(- \frac{p\xi}{4 e^u} \right) \widehat{\varphi}_{-p}^m (\xi) e^{-\zeta \xi} d\xi ,
\end{equation}
where $\mathcal{H}_0$ is the Hankel contour $\cH$ translated to 0. This in turn is equal to
\begin{equation}
-2 \ii (-1)^m e^{-\zeta \frac{4 e^u}{p}} \zeta  \int_{C^{\theta_{\textrm{St}}}} \widehat{\varphi}_{-p}^m (\xi) e^{-\zeta \xi} d\xi.
\end{equation}
Therefore, the monodromy in $\zeta$ of $s(\varphi_p^m)$ is given by 
\begin{equation}  \label{eqn:LateralBorel}
s(\varphi_p^m) (e^{2\pi \ii}\zeta) -s(\varphi_p^m) (\zeta) 
= 2\ii (-1)^{m+1} e^{-\zeta \frac{4 e^u}{p}} s_{\theta_{\textrm{St}}} (\varphi_{-p}^m)(\zeta). 
\end{equation}

Let us now explicitly compute the function $s(\varphi_p^m)(\zeta)$. Letting $\chi = t e^{i\theta}$, $t\in \R_+$ and using the Laplace transform formula~\eqref{eqn:hyperLaplace}, we have
\begin{equation}
\begin{split}
s_\theta (\varphi^m_p) (\zeta) 
&= \frac{1}{\sqrt{\pi}} e^{i \theta} \zeta \int_{0}^{\infty} 
{}_2F_1 \left( \frac12 - m, \frac12 + m, 1; \frac{e^{i \theta} p}{4e^u} t \right) e^{- e^{i \theta} \zeta t } dt \\
&= \frac{2}{\ii \pi} \sqrt{\frac{e^u}{p}} \ e^{-\zeta \frac{2 e^u}{p}} \zeta^{\frac12} K_m \left( -\zeta \frac{2 e^u}{p} \right),
\end{split}
\end{equation}
where $K_m(z)$ is the modified Bessel function of the second kind, see Appendix~\ref{sec:bessel}. 
Clearly this identity extends to the analytic continuations of the functions on the plane cut at $e^{-i \theta_{\textrm{St}}} \mathbb{R}_{+}$, therefore we have
\begin{align}
s (\varphi^m_p) (\zeta) = \frac{2}{\ii \pi} \sqrt{\frac{e^u}{p}} \ e^{-\zeta \frac{2 e^u}{p}} \zeta^{\frac12} K_m \left( -\zeta \frac{2 e^u}{p} \right). \label{eqn:actualsvarphi}
\end{align}

\begin{remark}
Equation~\eqref{eqn:actualsvarphi} is actually an identity between functions defined on the universal covering of $\C^*$ and formula~\eqref{eqn:asymptoticKn} for the asymptotics of $K_n(z)$ on a sector of opening $3 \pi$ induces the asymptotic formula~\eqref{eq:asp}.
\end{remark}

Let us now define the resummed weak functionals $ds_p(\zeta) \in T^*_{\hat{\lambda}_0} M^{\mathrm{weak}}$ for $p \in S^1$ by replacing in~\eqref{eqn:dyformal1}-\eqref{eqn:dyformal4} the formal series $\varphi_p^m$ with their Borel resummations $s(\varphi_p^m)$:
	\begin{align}
	\langle ds_p(\zeta), e_m \rangle &= \frac{1}{\ii \pi} \left( \frac{e^u}{p^2} + 1 \right) p^m e^{\zeta v} \zeta^{1/2} K_m \left( - \zeta \frac{2 e^u}{p} \right), \qquad &m \geq 1 \label{dspzm1} \\ 
	\langle ds_p(\zeta), e_m \rangle &= \frac{1}{\ii \pi}  \frac{e^u}{p^2} p^m e^{\zeta v} \zeta^{1/2} K_m \left( - \zeta \frac{2 e^u}{p} \right), \qquad &m \leq 0 \label{dspzm0} \\ 
	\langle ds_p(\zeta), e_v \rangle &= \frac{1}{\ii \pi} e^{\zeta v} \zeta^{1/2} K_0 \left( - \zeta \frac{2 e^u}{p} \right), \label{dsp010} \\ 
	\langle ds_p(\zeta), e_u \rangle &= \frac{1}{\ii \pi}  \frac{e^u}{p} e^{\zeta v} \zeta^{1/2} K_{1} \left( - \zeta \frac{2 e^u}{p} \right). \label{dsp001} 
	\end{align}
One can easily check that the weak functionals $ds_p(\zeta)$ solve the Dubrovin equation by a direct computation using the formulas for the derivatives of $K_m$ given in Appendix~\ref{sec:bessel}.

By replacing~\eqref{eqn:monodromyKnKn} (or, equivalently,~\eqref{eqn:LateralBorel}) in~\eqref{dspzm1}--\eqref{dsp001}, one can easily compute the monodromy of the weak functionals $ds_p$.
We summarize the results proved so far in the following proposition:
\begin{proposition} \label{prop:dsp}
	The weak functionals $ds_p (\zeta)$ solve the Dubrovin equation~\eqref{eqn:generalizedDubrovin}, and satisfy
	\begin{align}
	ds_{p} (\zeta) \sim dy_p^{\textrm{formal}} (\zeta), \qquad \text{for} \qquad  |\arg \zeta + \theta_0| < \frac{3 \pi}{2},
	\end{align}
	where $\theta_0 = \pi + \arg e^u -\arg p$. Their monodromy is given by
	\begin{align}
ds_p(\zeta e^{2 \pi \ii}) &= ds_p(\zeta) - 2ds_{-p}(\zeta), \label{eqn:monodromydsp} \\
ds_{-p}(\zeta e^{-2 \pi \ii}) &= ds_{-p} (\zeta ) - 2ds_p(\zeta), \label{eqn:monodromyds-p}
\end{align}
where $ds_{-p} = ds_{e^{-i \pi} p}$.
\end{proposition}

\begin{remark}
The solution above is actually multivalued in the parameter $p$. 
We will see that any choice of range $[ \phi_0, \phi_0 + 2\pi )$ for $\arg p$ gives a complete family of solutions. 
Notice that the family of solutions with $\arg e^u - \theta \leq \arg p \leq \arg e^u - \theta + 2\pi$  will have the formal asymptotics for $\zeta$ in the open half-plane $\Pi_\theta$, see~\eqref{eq:pt}.
\end{remark}

\begin{remark}
The asymptotic expansion~\eqref{eqn:Kmdoesnotconverge} of $K_m(z)$ for $m \to \pm \infty$ implies that 
\begin{equation}
|\langle ds_p(\zeta), e_m \rangle| \sim \beta_\pm |m|^{-1/2} (\alpha_\pm |m|)^{|m|} 
\end{equation}
for positive constants $\alpha_\pm$, $\beta_\pm$.
Therefore the weak functionals $ds_p(\zeta)$ defined by the coefficients above do not extend to cotangent vectors in $T^*_{\hat{\lambda}_0} M$.
\end{remark}

While the weak functional $ds_p(\zeta)$ do not define elements in the dual to $T_{\hat{\lambda}_0} M$, i.e., they are not cotangent vectors, the difference $ds_p - ds_{-p}$ is not only an element of $T^*_{\hat{\lambda}_0} M$, but is actually representable. More precisely,

\begin{proposition} 
For $\sigma = \sigma(p)$, we have that
\begin{align}
dy_\sigma(\zeta) = ds_p (\zeta) - ds_{-p} (\zeta).
\end{align}
\end{proposition}

\begin{proof}
The coefficients of the integral solutions defined in Section~\ref{sec:analyticsolutions}, i.e.,
\begin{align}
\langle dy_\sigma, \hat{X} \rangle = \frac{\zeta^{1/2}}{2 \pi \ii} \oint_{|z| = 1} e^{\zeta \lambda_\sigma(z)} \langle d\lambda_\sigma(z), \hat{X} \rangle \frac{dz}{z}, \label{eqn:integralsolutions67}
\end{align}
obtained by acting on the elements $e_{\hat{m}}$ where $\hat{m} \in \Z \cup \{v, u\}$, see~\eqref{eq:coeffs}, are given by Bessel functions of the first kind
\begin{align}
\langle dy_\sigma, e_m \rangle &= \left( \frac{e^u}{p^2} + 1 \right) \zeta^{\frac12} e^{\zeta v} p^m I_m \left( \zeta \frac{2 e^u}{p} \right), \qquad &m \geq 1 \label{eqn:dysigmamgeq1} \\
\langle dy_\sigma, e_m \rangle &= \frac{e^u}{p^2} \zeta^{\frac12} e^{\zeta v} p^m I_m \left( \zeta \frac{2 e^u}{p} \right), \qquad &m \leq 0 \label{eqn:dysigmamleq0} \\
\langle dy_\sigma, e_v \rangle &= \zeta^{\frac12} e^{\zeta v} I_0 \left( \zeta \frac{2 e^u}{p} \right), \label{eqn:dysigma010}\\
\langle dy_\sigma, e_u \rangle &=\frac{e^u}{p} \zeta^{\frac12} e^{\zeta v}  I_1 \left( \zeta \frac{2 e^u}{p} \right), \label{eqn:dysigma001}
\end{align}
where $p = p(\sigma)$. Let us illustrate how to obtain the coefficients~\eqref{eqn:dysigmamgeq1}--\eqref{eqn:dysigma001} by proving~\eqref{eqn:dysigmamgeq1}. Let $m \geq 1$, then 
\begin{align}
\langle dy_\sigma, e_m \rangle &= \sigma e^{\zeta v} \zeta^{\frac12} \frac{1}{2 \pi \ii} \oint_{|z| = 1} e^{\zeta \frac{e^u}{p} \left( \frac{z}{p} + \frac{p}{z} \right)} z^m \frac{dz}{z} \\
&= \left(\frac{e^u}{p^2} + 1 \right) e^{\zeta v} \zeta^{\frac12} p^m \frac{1}{2 \pi \ii} \oint_{|w| = 1} e^{\frac12 \zeta \frac{2 e^u}{p} \left( w + \frac{1}{w} \right)} w^m \frac{dw}{w} \notag
\end{align}
where in the second line we replaced $w= z/p$.
Equation~\eqref{eqn:dysigmamgeq1} follows by noting that the integral is a residue of the generating function~\eqref{eqn:generatingIn}. The proposition follows by applying the monodromy identity~\eqref{eqn:monodromyKnIn} to~\eqref{eqn:dysigmamgeq1}--\eqref{eqn:dysigma001} and~\eqref{dspzm1}--\eqref{dsp001}.
\end{proof}

Despite the fact that the weak functionals $ds_p(\zeta)$ do not extend to $T_{\hat{\lambda}_0} M$ as explained above, we can still ask the question about their completeness as a family of functionals on $T_{\hat{\lambda}_0} M^{\mathrm{test}}$ for fixed $\zeta$. It is indeed the case that the map 
\begin{equation} \label{eq:xdspz}
\hat{X} \longmapsto \langle ds_p(\zeta), \hat{X} \rangle,
\end{equation}
that associates to $\hat{X} \in T_{\hat{\lambda}_0} M^{\mathrm{test}}$ a function of $p$ with $|p|=1$ and $\arg p \in [ \phi_0, \phi_0 + 2\pi )$ for some fixed $\phi_0$ is injective, as proved in the following

\begin{proposition} \label{prop:dspiscomplete}
The map~\eqref{eq:xdspz} associated with the family of functionals $\{ds_p(\zeta)\}$ is injective.
\end{proposition}

\begin{proof}
 Let $\hat{X} = (X(z), X_v, X_u) \in T_{\hat{\lambda}_0} M^{\mathrm{test}}$ for 
\begin{align}
X(z) = X_{-s} z^{-s} + \dots + X_{-1} z^{-1} + X_0 + X_1 z + \dots X_r z^r \in \C[z, z^{-1}].
\end{align}
Then from~\eqref{dspzm1}--\eqref{dsp001} we get
\begin{align}
\ii \pi e^{\zeta v} \zeta^{-\frac12} \langle ds_p (\zeta), \hat{X} \rangle &= X_{-s} e^u p^{-s-2} K_s \left( - \zeta \frac{2 e^u}{p} \right) + \dots + X_0 e^u p^{-2}  K_0 \left( -\zeta \frac{2 e^u}{p} \right) \label{eqn:provingcompleteness} \\
&+ X_1 \left( e^u p^{-2} + 1 \right) p K_1 \left( -\zeta \frac{2 e^u}{p} \right) + \dots + X_r \left( e^u p^{-2} + 1 \right) p^r K_r \left( -\zeta \frac{2 e^u}{p} \right) \notag \\ 
&+ X_v K_0 \left( -\zeta \frac{2 e^u}{p} \right) + X_u e^u p^{-1} K_1 \left( -\zeta \frac{2 e^u}{p} \right) \notag .
\end{align}
By~\eqref{eqn:multivaluedKn}, the expression above is an expansion in $\{ p^{2m}, p^{2n} \log (- \zeta e^u p^{-1})  \}_{n, m \in \Z}$. Assume $\langle ds_p (\zeta), \hat{X} \rangle = 0$ for all $p$. To show completeness, we need to prove $\hat{X} = 0$. By~\eqref{eqn:multivaluedKn}, the coefficient of $p^{2r}$ of~\eqref{eqn:provingcompleteness} equals
\begin{align}
X_r \frac12 (-\zeta e^u)^{-r} (r-1)!,
\end{align}
which must be zero, so $X_r = 0$. Repeating this argument with the coefficients of $p^{2r-2}, \dots, p^2$ shows $X_{r-1} = \dots = X_1 = 0$. We are left with 
\begin{align}
 &X_{-s} e^u p^{-s-2} K_s \left( - \zeta \frac{2 e^u}{p} \right) + \dots + X_0 e^u p^{-2}  K_0 \left( -\zeta \frac{2 e^u}{p} \right) + \label{eqn:provingcompleteness2}\\
 &+ X_v K_0 \left( -\zeta \frac{2 e^u}{p} \right) + X_u e^u p^{-1} K_1 \left( -\zeta \frac{2 e^u}{p} \right) = 0. \notag
\end{align}
The coefficient of $\log(- \zeta e^u p^{-1})$ of~\eqref{eqn:provingcompleteness2} equals $-X_v$, so $X_v = 0$. The constant coefficient equals
\begin{align}
e^u X_u \left( -\zeta 2 e^u \right)^{-1},
\end{align} 
so $X_u = 0$. Repeating this argument with the coefficients of $\log(- \zeta e^u p^{-1}) p^{-2}, \dots, \log(- \zeta e^u p^{-1}) p^{-2s-2}$ shows that $X_{0} = \dots = X_{-s} = 0$.
\end{proof}

\subsection{Stokes matrices for pairs of solutions} 
\label{sec:Stokesoperators}

In this section, we restrict to pairs of solutions and we compute the partial Stokes matrix that describes their monodromy. 

Let $ds_p$ and $ds_{-p}$ be the solutions corresponding to arguments $\arg p$ and $\arg p - \pi$, respectively. 
Recall their formal asymptotics as $|\zeta|\to \infty $
\begin{align}
ds_{p}(\zeta) &\sim dy_p^{\mathrm{formal}} = e^{\zeta u_p} (r_p^0 + r_p^1 \zeta^{-1} + \dots ), \qquad  &\arg \zeta \in ( -\theta_0 - \frac{3\pi}{2}, -\theta_0 + \frac{3\pi}{2} ) ,\\
ds_{-p}(\zeta) &\sim dy_{-p}^{\mathrm{formal}} = e^{\zeta u_{-p}} (r_{-p}^0 + r_{-p}^1 \zeta^{-1} + \dots ),  \qquad 
&\arg \zeta \in ( -\theta_0 - \frac{5\pi}{2}, -\theta_0 + \frac{\pi}{2} ),
\end{align}
where $dy_p^{\mathrm{formal}}$ is given by~\eqref{eq:frml}, and $\theta_0 = \pi + \arg e^u - \arg p$. 

The Stokes line $\ell_{St}$ separates the two halves of the complex plane where $e^{\zeta u_p}$ and $e^{\zeta u_{-p}}$ are respectively dominant for $|\zeta| \to \infty $. It is given by 
\begin{align} \label{eqn:Stokesline}
\ell_{\textrm{St}} =  \left\{ \zeta \in \C\ \Big| \ \Re \left( \zeta u_p \right) = \Re \left(\zeta u_{-p} \right)  \right\}, \end{align}
namely the line of argument $\theta_0 +\frac{\pi}{2} \ \mathrm{mod}\ \pi$.  Notice that the exponential $e^{\zeta u_p}$ dominates $e^{\zeta u_{-p}}$ if $\arg \zeta \in ( -\theta_0 + \frac{\pi}{2}, -\theta_0 + \frac{3\pi}{2} )$.

We choose an admissible line $\ell$ not coinciding with the Stokes line, in this case the positive direction of $\ell$ is of argument $\theta$ with $\theta \not=\theta_0 +\frac{\pi}{2} \ \mathrm{mod}\ \pi$.

For a small $\epsilon>0$, we define two sectors containing the half-planes separated by $\ell$ as follows
\begin{align}
	\Pi^\epsilon_{\textrm{right}}  = \{ \zeta \in \C \ | \ \theta - \pi - \epsilon < \arg \zeta < \theta + \epsilon \}, \\
	\Pi^\epsilon_{\textrm{left}}  = \{ \zeta \in \C \ | \ \theta - \epsilon < \arg \zeta < \theta + \pi + \epsilon \}.
\end{align}
The intersection of $\Pi^\epsilon_{\textrm{right}} $ and $\Pi^\epsilon_{\textrm{left}}$ has two connected components 
\begin{align}
\Pi^\epsilon_{+} &= \{ \zeta \in \C \ | \ \theta - \epsilon < \arg \zeta < \theta + \epsilon \},\\
\Pi^\epsilon_{-}  &= \{ \zeta \in \C \ | \ \theta+\pi - \epsilon < \arg \zeta < \theta + \pi + \epsilon \} .
\end{align}
Let us assume that the argument $\theta$ of the admissible line $\ell$ has been chosen in such a way that $ds_{\pm p}(\zeta)$ is dominant in $\Pi_{\mp}^\epsilon$; this amounts to $\theta \in ( -\theta_0 - \frac{\pi}{2}, -\theta_0 + \frac{\pi}{2} )$.

Let us define the following ``matrix'' solutions on $\Pi^\epsilon_{\textrm{right/left}}$ 
\begin{align}
Y_{\textrm{right}}(\zeta) &=  \left( ds_p(\zeta), ds_{-p}(\zeta) \right),
\qquad \qquad \theta - \pi - \epsilon < \arg \zeta < \theta + \epsilon, \\
Y_{\textrm{left}}(\zeta) &=  \left( ds_p(\zeta), ds_{-p}(\zeta e^{-2\pi \ii}) \right),
\qquad  \theta - \epsilon < \arg \zeta < \theta + \pi + \epsilon ,
\end{align}
where we have chosen the appropriate branch cuts that guarantee the formal asymptotics in the half-plane where they are defined.

\begin{theorem} \label{thm:Stokes} 
The solutions $Y_{\textrm{right}}(\zeta)$ and $Y_{\textrm{left}}(\zeta)$ defined above have the formal asymptotics
	\begin{align}
	Y_{\textrm{left/right}} (\zeta) \sim \begin{pmatrix}
	dy_p^{\textrm{formal}} (\zeta) , & dy_{-p}^{\textrm{formal}} (\zeta)
	\end{pmatrix}
	\end{align}
for $|\zeta| \to \infty$ in their respective domains of definition $\Pi^\epsilon_{\textrm{right/left}}$.
On their common domains of definition $\Pi^\epsilon_\pm$ they are related by
	\begin{align}
	Y_{\textrm{left}} (\zeta) &= Y_{\textrm{right}} (\zeta) S_{+}, \qquad \zeta \in \Pi_{+}^\epsilon, \\
	Y_{\textrm{left}} (\zeta) &= Y_{\textrm{right}} (\zeta) S_{-}, \qquad \zeta \in \Pi_{-}^\epsilon,
	\end{align}
where the Stokes matrices $S_\pm$ are given by
	\begin{align}
	S_{-} = \begin{pmatrix}
	1 & 0 \\ -2 & 1
	\end{pmatrix}, \qquad 
	S_{+} = \begin{pmatrix}
	1 & -2 \\ 0 & 1
	\end{pmatrix}.
	\end{align}
\end{theorem}

\begin{proof}
The theorem follows from Proposition~\ref{prop:dsp}. 
\end{proof}

\subsection{Stokes matrices}

Let us fix $\theta \in \R$  and define two open half-planes  $\Pi_{\mathrm{right/left}}$ as  follows
\begin{align}
	\Pi_{\textrm{right}}  = \{ \zeta \in \C \ | \ \theta - \pi  < \arg \zeta < \theta  \}, \\
	\Pi_{\textrm{left}}  = \{ \zeta \in \C \ | \ \theta  < \arg \zeta < \theta + \pi  \}.
\end{align}

Let us define two families of solutions $Y_{\mathrm{right}}$ and $Y_{\mathrm{left}}$ of the Dubrovin equation with formal asymptotics in the half-planes $\Pi_{\mathrm{right/left}}$ respectively. These can be seen as the analogues of the fundamental solutions in the finite-dimensional case.

The family  $Y_{\mathrm{right}}$ is defined on $\Pi_{\mathrm{right}}$ by
\begin{equation}
\left( Y_{\mathrm{right}}(\zeta)\right)_p = 
ds_p(\zeta)  \quad \text{for} \quad  \arg p \in 
[ \arg e^u + \theta -\frac{\pi}{2} , \arg e^u + \theta +\frac{3\pi}{2}),
\end{equation}
where $\theta-\pi < \arg \zeta < \theta$; the family $Y_{\mathrm{left}}$ is defined on $\Pi_{\mathrm{left}}$ by
\begin{equation}
\left( Y_{\mathrm{left}}(\zeta)\right)_p 
= \begin{cases}
ds_p(\zeta), & \arg p \in  (\arg e^u + \theta + \frac{\pi}{2},  \arg e^u + \theta + \frac{3\pi}{2}  )     \\
ds_p(e^{-2\pi \ii} \zeta), & \arg p \in (\arg e^u + \theta - \frac{\pi}{2} ,  \arg e^u + \theta + \frac{\pi}{2}  ) 
\end{cases},
\end{equation}
where $\theta < \arg \zeta < \theta +\pi$.

While the fundamental solutions $Y_{\mathrm{right/left}}$ have formal asymptotics only in the domains $\Pi_{\mathrm{right/left}}$, they can be nevertheless analytically continued beyond those sectors and therefore compared, defining operators that are  infinite-dimensional analogues of the Stokes matrices. We summarize these observations and we compute the Stokes operators in the following theorem. 

\begin{theorem}
The families of solutions $Y_{\textrm{right}}$ and $Y_{\textrm{left}}$ have the formal asymptotics
\begin{equation}
\left( Y_{\mathrm{right/left}}(\zeta)\right)_p \sim dy^{\mathrm{formal}}_p(\zeta)
\end{equation}
for $|\zeta| \to \infty$ in the half-planes $\Pi_{\textrm{right/left}}$.

On the sectors $\Pi^\epsilon_\pm$ they are related by
\begin{equation}
\left( Y_{\mathrm{left}}(\zeta)\right)_p =\left( Y_{\mathrm{right}}(\zeta)\right)_p
-2 \begin{cases}
0,
& \arg p \in (\arg e^u + \theta + \frac{\pi}{2} ,  \arg e^u + \theta + \frac{3\pi}{2} ) \\
\left( Y_{\mathrm{right}}(\zeta)\right)_{e^{\pi \ii}p},
& \arg p \in (\arg e^u + \theta - \frac{\pi}{2} ,  \arg e^u + \theta + \frac{\pi}{2} )
\end{cases}
\end{equation}
for $\zeta \in \Pi_{+}^\epsilon$, and
\begin{equation}
\left( Y_{\mathrm{left}}(\zeta)\right)_p =\left( Y_{\mathrm{right}}(\zeta)\right)_p
-2 \begin{cases}
\left( Y_{\mathrm{right}}(\zeta)\right)_{e^{-\pi \ii}p},
& \arg p \in (\arg e^u + \theta + \frac{\pi}{2} ,  \arg e^u + \theta + \frac{3\pi}{2} ) \\
0, & \arg p \in (\arg e^u + \theta - \frac{\pi}{2} ,  \arg e^u + \theta + \frac{\pi}{2} )
\end{cases}
\end{equation}
for $\zeta \in \Pi_{-}^\epsilon$.
\end{theorem}

\begin{remark}
We can formally express the relation between $Y_{\textrm{right}}$ and $Y_{\textrm{left}}$ in terms of kernels $S_\pm$ by writing
\begin{equation}
\left( Y_{\mathrm{left}}(\zeta)\right)_p = \int_{S^1} \left( Y_{\mathrm{right}}(\zeta)\right)_q (S_\pm)_{q p} dq ,
\end{equation}
where the integral is taken on the points $q$ in $S^1$ with argument in $[ \arg e^u + \theta -\frac{\pi}{2} , \arg e^u + \theta +\frac{3\pi}{2})$. 

The kernels representing the analogues of the Stokes matrices are then written as
\begin{align}
(S_+)_{qp} &= \delta(q-p) -2 \chi(q) \delta(q - e^{\pi \ii} p ) , \\
(S_-)_{qp} &= \delta(q-p) -2 \chi(p) \delta(p-e^{\pi \ii}q ),
\end{align}
where $\chi(p)$ is the function equal to one when $\arg p$ is in $(\arg e^u + \theta + \frac{\pi}{2} ,  \arg e^u + \theta + \frac{3\pi}{2} )$ and zero otherwise, and the delta function satisfies the usual relation
\begin{equation}
\int_{S^1} f(q) \delta(q-p) dq = f(p).
\end{equation}
Notice that the two kernels $S_+$ and $S_-$ are the transposes of one another, namely
\begin{equation}
(S_+)_{pq} = (S_-)_{qp}.
\end{equation}
\end{remark}

\appendix 

\section{Saddle point asymptotics}

Let us recall the proof of the following lemma, which can be seen as a simple application of Perron's method~\cite{Won89} to our particular case.

\begin{lemma} \label{lem:saddlepoint}
	Let $f$ and $g$ be holomorphic functions defined in a neighbourhood of a point $z^{\prime}$ where $f$ has a simple critical point and let $\cC$ be a path passing through $z^{\prime}$ such that the real part of $e^{\ii \psi} f(z)$ restricted to $\cC$ has a maximum at $z^{\prime}$. Then the function of $\zeta$ defined by 
	\begin{equation}
	\cI =  \zeta^{1/2} \int_{\cC} e^{\zeta f(z)} g(z) dz
	\end{equation}
	admits the asymptotic expansion
	\begin{equation}
	\cI \sim e^{\zeta f(z^{\prime})} \sum_{n \geq 0} d_n \zeta^{-n},  
	\end{equation}
	for $\zeta = |\zeta| e^{\ii \psi}$, $ |\zeta| \to +\infty$, with
	\begin{equation}
	d_n =\ii (-1)^n\Gamma(n+1/2) \res_{z=z^{\prime}} \frac{g(z)}{(f(z) - f(z^{\prime}))^{n+1/2}} dz.
	\end{equation}
\end{lemma}

\begin{proof}
	By shifting the variable of integration and renaming $f$ and $g$ we can assume that $f$ and $g$ are analytic in a neighborhood of $z=0$  with $f(0) = f'(0) =0$ and $f''(0) \not=0$. We write $f(z) = cz^2 + O(z^3)$ with $c = \frac{f''(0)}{2} \in \C^*$.
	
	By deforming the path $\cC$ we can make it coincide with a steepest descent path in a sufficiently small neighborhood of the critical point. We can moreover restrict the integral to a part of the path arbitrarily close to the critical point without changing the asymptotic expansion, as the difference will be exponentially vanishing. 
	
	We will therefore assume that $\cC$ is steepest descent path defined as the preimage of the path $\chi(t) = - e^{-\ii \psi} t$ for $t \in [0,T]$ via $f(z)$ with the appropriate orientation. Denote by $\cC_+$ the part of the path $\cC$ leaving the critical point and by $\cC_-$ the one arriving at the critical point. 
		
	Let $w(z)$ be the unique square root of $c^{-1} f(z)$ with $w(z)  =z + O(z^2)$. 
	The function $w(z)$ is biholomorphic, so we can use it to change variable of integration; denoting $z(w)$ the inverse, we get
	\begin{equation}
	\cI = \zeta^{1/2} \int_{\cC} e^{\zeta c w^2} s(w) dw ,
	\end{equation}
	where $s(w) = \frac{g(z(w))}{w'(z(w))}$ is holomorphic at $w=0$ with Taylor expansion $s(w) = \sum_{n \geq0} s_n w^n$.
	
	Let $\tilde{\cC}$ be the path $\eta(t) = - e^{-\ii \psi} c^{-1} t$ for $t \in [0,T]$. Let $\sqrt{\eta}$ be the branch of the square root that maps  $\tilde{\cC}$ to $\cC_+$ (we choose a branch cut for the square root that does not coincide with $\tilde{\cC}$). The other branch $-\sqrt{\eta}$ maps  $\tilde{\cC}$ to $-\cC_-$. Splitting the integral in the two parts corresponding to $\cC_+$ and $\cC_-$ and changing the variable of integration with $w = \sqrt{\eta}$ and $w = -\sqrt{\eta}$ respectively, we obtain
	\begin{equation}
	\cI= \zeta^{1/2} \int_{\tilde{\cC}} e^{\zeta c \eta} \tilde{s}(\eta) \frac{d \eta}{\sqrt{\eta}},
	\end{equation}
	where 
	\begin{equation}
	\tilde{s}(\eta) = \frac12 \left(s(\sqrt{\eta}) + s(-\sqrt{\eta}) \right)
	= \sum_{n\geq0} s_{2n} \eta^{n}.
	\end{equation}
	The integral is explicitly given by
	\begin{equation}
	\cI = \zeta^{1/2} \int_0^T e^{-|\zeta| t} t^{-1/2} a(t) dt,
	\end{equation}
	with $a(t) = \sqrt{-\frac{e^{-\ii \psi}}{c} } \tilde{s} \left( -\frac{e^{-\ii \psi}}{c} t \right)$.

	According to Watson's Lemma (see Proposition~2.1 in~\cite{Mil06}) we have the following asymptotic expansion as $|\zeta| \to \infty$
	\begin{equation}
	\int_0^T e^{-|\zeta| t} t^{-1/2} a(t) dt \sim
	\sum_{n \geq0} \Gamma(n+1/2) \frac{a^{(n)}(0)}{n!} |\zeta|^{-n-1/2},
	\end{equation}
	for any complex valued smooth function $a(t)$ defined in a neighborhood of $[0,T]$.
	Clearly $a^{(n)}(0) = n!   s_{2n} \left( -\frac{e^{-\ii \psi}}{c}  \right)^{n+1/2}$, so we obtain the asymptotic expansion
	\begin{equation}
	\cI \sim \ii \sum_{n \geq 0} \Gamma(n+1/2) \frac{s_{2n}}{c^{n+1/2}} (- \zeta)^{-n}.
	\end{equation}
	
	We can finally compute the coefficients $s_n$ as residues
	\begin{equation}
	s_n = \res_{w=0} \frac{g(z(w))}{w'(z(w))} \frac{dw}{w^{n+1}} = 
	\res_{z=0} \frac{g(z)}{w(z)^{n+1}} dz.
	\end{equation}
	Expressing $w(z)$ as square root of $c^{-1} f(z)$ we obtain the desired result. 
\end{proof}
\begin{remark}
	We choose the branches of the roots of $c$ and $e^{\ii \psi}$ such that the sign in the final expression is $+1$.
\end{remark}

\section{Special functions}

\subsection{Modified Bessel functions} 
\label{sec:bessel}

In this appendix, we go over the definition and some properties of the modified Bessel functions. For more details, we refer the reader to~\cite[Sections 10.25-10.46]{DLMF}. The modified Bessel functions of the first kind are defined by
\begin{align}
I_\nu (z) &= \sum_{k=0}^{\infty} \frac{1}{\Gamma \left( k + \nu + 1 \right) k! } \left( \frac{z}{2} \right)^{2k + \nu}.
\end{align}
The modified Bessel functions of the second kind are defined by
\begin{align}
K_\nu (z) &= \frac{\pi}{2} \frac{I_{-\nu} (z) - I_\nu(z)}{\sin (\pi \nu)}, \qquad \nu \notin \mathbb{Z}, \\
K_m (z) &= \lim_{\mu \rightarrow m} K_\mu(z), \qquad m \in \Z.
\end{align}
For $n \in \Z$, $I_n(z)$ is entire and $K_n(z)$ is multivalued with a branch cut on $\R_{-}$. Its multivaluedness becomes clear from the expansion at $z=0$
\begin{align}
K_n(z) &= \frac12 \left( \frac{z}2  \right)^{-n} \sum_{k=0}^{n-1} \frac{(n-k-1)!}{k!} \left( - \frac14 z^2 \right)^k + (-1)^{n+1} \log \left( \frac{z}2 \right) I_n(z) \label{eqn:multivaluedKn}  \\
&+ (-1)^n \frac12 \left( \frac{z}2  \right)^n \sum_{k=0}^{\infty} \left( \psi(k+1) + \psi(n+k+1) \right) \frac{\left( \frac{z}2 \right)^{2k}}{k! (n+k)!}, \notag
\end{align}
where
\begin{align}
\psi(z) &= \sum_{k=1}^{\infty} \left( \frac1k - \frac1{k+z-1} \right) - \gamma, \label{eqn:defnpsi} \\
\gamma &= \lim_{n \rightarrow \infty} \left( \sum_{k=1}^{n} \frac1k - \log n \right). \label{eqn:defngamma}
\end{align}
The following properties will be used in the text: 
\begin{align}
I_n(-z) &= (-1)^n I_n(z), \label{eqn:parityIn} \\
I^{\prime}_n(z) &= I_{n-1}(z) - \frac{n}{z} I_n(z) \label{eqn:derivativeIn}, \\
K^{\prime}_n(z) &= - K_{n-1}(z) - \frac{n}{z} K_n(z) \label{eqn:derivativeKn}.
\end{align}
The monodromy of $K_n(z)$ is given by
\begin{align}
K_n(z e^{m\pi\ii}) &= (-1)^{mn} K_n(z) - (-1)^{n(m-1)} m \pi \ii I_n(z), \label{eqn:monodromyKnIn} \\
K_n(z e^{m\pi \ii}) &= (-1)^{n(m-1)} m K_n(ze^{\pi \ii}) - (-1)^{nm} (m-1) K_n(z). \label{eqn:monodromyKnKn}
\end{align}
It is also useful to keep in mind their asymptotic expansions for large $n$
\begin{align}
I_n(z) \sim \frac{1}{\sqrt{2 \pi n}} \left( \frac{ez}{2n} \right)^n, \label{eqn:Indoesconverge} \\
K_n(z) \sim \sqrt{\frac{\pi}{2 n}} \left( \frac{2n}{ez}  \right)^n, \label{eqn:Kmdoesnotconverge}
\end{align}
and for large $z$
\begin{align}
I_n(z) &\sim \frac{e^z}{\sqrt{2 \pi z}} \sum_{k=0}^{\infty} (-1)^k a_k(n) z^{-k}, \qquad &|\arg z| < \frac12 \pi, \ |z| \rightarrow \infty, \label{eqn:asymptoticIn} \\
K_n(z) &\sim e^{-z} \sqrt{\frac{\pi}{2z}} \sum_{k=0}^{\infty} a_k(n) z^{-k}, \qquad &|\arg z| < \frac32 \pi, \ |z| \rightarrow \infty, \label{eqn:asymptoticKn}
\end{align}
where $a_0(n) = 0$, and
\begin{align}
a_k(n) = \frac{(4n^2-1^2)(4n^2-3^2) \dots (4n^2 - (2k-1)^2) }{k! 8^k}.
\end{align}
The $I_n$ can be encoded together in a generating function
\begin{align}
e^{\frac12 z \left( t + t^{-1} \right) } = \sum_{n=-\infty}^{\infty} t^n I_n(z), \label{eqn:generatingIn}
\end{align}
which converges for all $t \in \mathbb{C}^*$.

\subsection{Gauss hypergeometric functions} 
\label{sec:hypergeometric} 

The definition and properties of the Gauss hypergeometric functions presented below are taken from~\cite[Chapter 15]{DLMF}. For more details, we refer the reader to that source.
The Gauss hypergeometric function is defined by the power series
\begin{align}
{}_2F_1(a, b; c; z) = \sum_{n=0}^{\infty} \frac{(a)^{n} (b)^{n} }{n! (c)^{n} } z^n, \label{eqn:hypergeometricdefn}
\end{align}
in the disk $|z| < 1$ and by analytic continuation elsewhere, where 
\begin{align}
(q)^n = q(q+1)\dots (q+n-1)
\end{align}
denotes the rising factorial. At $z=1$, they have a logarithmic branch point of the form
\begin{align}
{}_2F_1(a, b, a+b; z) &= - \frac{\Gamma(a+b)}{\Gamma(a) \Gamma(b)} \log(1-z) \ {}_2F_1(a, b, a+b; 1-z) \label{eqn:hypergeometriclog} \\
&+ \sum_{k=0}^{\infty} \frac{(a)^k (b)^k}{k!^2} \left( 2 \psi(k+1) - \psi(a+k) - \psi(b+k) \right) (1-z)^k, \nonumber
\end{align}
where the function $\psi$ is defined by~\eqref{eqn:defnpsi}--\eqref{eqn:defngamma}. 

The following Laplace transform will be used in the text, see~\cite{PBM92}: for $\Re c, \Re q > 0, |\arg \omega| < \pi$, we have
\begin{align}
\int_{0}^{\infty} {}_2F_1(a,1-a;c; - \omega x) e^{-qx}dx = \frac{q^{\frac12-c}}{\sqrt{\pi \omega}} \Gamma(c) e^{\frac{q}{2 \omega}} K_{a - \frac12} \left( \frac{q}{2 \omega} \right), \label{eqn:hyperLaplace}
\end{align}
where $K_\nu (z)$ is the modified Bessel function of the second kind, defined in Appendix~\ref{sec:bessel}.

Finally, we state and prove a technical lemma necessary for the Borel resummation procedure performed in Section~\ref{sec:resurgenceandstokes}. 
\begin{lemma} \label{lemma:aux-hg}
	For $|z| < 4$, the power series 
	\begin{align}
	f(z) = \sum_{k=0}^\infty \binom{-1/2}{k} \binom{m+k-1/2}{2k} z^k
	\end{align}
	converges and coincides with the Gauss hypergeometric function
	\begin{align}
	f(z) = {}_2F_1 \left(\frac12 - m, \frac12 + m, 1; \frac{z}{4} \right).
	\end{align} 
\end{lemma}
\begin{proof}
	We write
	\begin{align}
	\binom{-\frac12} {k} &= \frac{(-1)^k (\frac12)^{(k)} }{k!}, \\
	\binom{m+k-\frac12}{2k} &= \frac{ (m+\frac12)^{(k)} (m-\frac12)_{(k)} }{(2k)!},
	\end{align}
	where $(a)^{(n)}$ is the rising factorial and $(a)_{(n)}$ is the falling factorial, given by
	\begin{align}
	(a)_{(n)} &= a(a-1)(a-2) \dots (a-n+1), \\
	(a)^{(n)} &= a(a+1)(a+2) \dots (a+n-1).
	\end{align}
	Using the property
	\begin{align}
	(a)^{(n)} = (-1)^n (-a)_{(n)},
	\end{align}
	we can write
	\begin{align}
	f(z) = \sum_{k=0}^\infty \frac{ (\frac12)^{(k)} (m+\frac12)^{(k)} (\frac12 - m)^{(k)} }{(2k)! k!} z^k.
	\end{align}
	Comparing this expression to the power series expansion~\eqref{eqn:hypergeometricdefn} of the hypergeometric function reduces the lemma to proving the identity
	\begin{align}
	4^k k! (1/2)^{(k)} = (2k)!, \qquad k = 0, 1, \dots,
	\end{align}
	which, after noting $(2k)! = k! (k+1)^{(k)}$, further reduces to
	\begin{align}
	(k+1)^{(k)} = 4^k (1/2)^{(k)}. \label{eqn:identityrisingfactorial}
	\end{align}
	To prove~\eqref{eqn:identityrisingfactorial}, we proceed via induction. For $k = 0, 1$ it is clear that it holds. Assume it is true for $k \geq 1$. Then
	\begin{align}
	(k+2)^{(k+1)} &= (k+2)^{(k)} (2k+2) = \frac{(k+1)^{(k+1)}}{k+1} (2k+2) = 2 (k+1)^{(k+1)}= \\
	&= 2(k+1)^{(k)} (2k+1) = 4^k (1/2)^{(k)} 2 (2k+1) = \notag \\ 
	&= 4^{k+1} (1/2)^{(k)} (k+1/2) = 4^{k+1} (1/2)^{(k+1)}\notag .
	\end{align}
\end{proof}


\bibliography{\jobname}
\bibliographystyle{amsalpha}

\end{document}